\documentclass[journal]{IEEEtran}

\usepackage[T1]{fontenc}
\usepackage{amsmath}
\usepackage{amsthm}
\usepackage[cmintegrals]{newtxmath}
\usepackage{amssymb}
\usepackage{bm}
\usepackage{bbm}
\usepackage{comment}
\usepackage{graphicx}
\usepackage{calc}
\usepackage{epstopdf}
\newtheorem{theorem}{Theorem}
\newtheorem{lemma}{Lemma}
\newtheorem{remark}{Remark}

\usepackage{verbatim}
\usepackage{color}
\usepackage{array}
\usepackage{cite}
\usepackage{stfloats}

\usepackage{xcolor}
\usepackage{floatrow}
\usepackage{lipsum}
\usepackage[ruled,norelsize]{algorithm2e}
\usepackage{algorithmic}
\usepackage[caption=false,subrefformat=parens,labelformat=parens]{subfig}
\usepackage{multirow}
\begin{document}
\title{Learning Based Hybrid Beamforming Design  for  Full-Duplex Millimeter Wave Systems}

\author{
	Shaocheng~Huang, Yu~Ye, \IEEEmembership{Student Member, IEEE}  and  Ming~Xiao, \IEEEmembership{Senior Member, IEEE}
	\thanks{S.~Huang, Y.~Ye and M.~Xiao are with the Division of Information Science and Engineering, KTH Royal Institute of Technology, Stockholm, Sweden (e-mail: \{shahua, yu9, mingx\}@kth.se).}
}

\maketitle

\begin{abstract}
Millimeter Wave (mmWave) communications  with full-duplex (FD) have the potential of increasing the spectral efficiency, relative to those with half-duplex. However, the residual self-interference (SI) from FD and high pathloss inherent to mmWave signals may degrade the system performance. 
Meanwhile, hybrid beamforming (HBF) is an efficient technology to enhance the channel gain and mitigate  interference with reasonable complexity.  However, conventional HBF approaches for FD mmWave systems are  based on optimization  processes, which are either too complex or  strongly  rely on the quality of channel state information (CSI). 
We propose two learning schemes to design HBF for FD mmWave systems,  i.e., extreme learning machine  based HBF (ELM-HBF) and convolutional neural networks based HBF (CNN-HBF). Specifically, we first propose an alternating direction method of multipliers (ADMM) based algorithm to achieve SI cancellation beamforming, and then use a majorization-minimization (MM) based algorithm for joint transmitting and receiving HBF optimization. 
To train the learning networks, we simulate  noisy channels as input, and select the hybrid beamformers calculated by  proposed algorithms as targets. Results show that both learning based schemes can provide more robust HBF performance and achieve at least $22.1\%$ higher spectral efficiency compared to orthogonal matching pursuit (OMP) algorithms. Besides, the online prediction time of proposed learning based schemes is almost 20 times faster than the OMP scheme. Furthermore,  the  training  time  of  ELM-HBF  is  about  600  times  faster  than  that  of CNN-HBF  with $64$ transmitting and receiving antennas. 



\end{abstract}


\section{Introduction}

With the development of various emerging applications (e.g., virtual reality, augmented reality, autonomous driving and big data analysis),  
data traffic has explosively increased and caused   growing demands for very high  communication rates  in future wireless communications, e.g.,  the fifth generation (5G) and beyond\cite{Ming2017xiao}.  A common approach to
meet the requirements of high  rates  is to explore the potential for improvements in bandwidth and spectral
efficiency. 
 Millimeter wave (mmWave) communications have recently
 received increasing research attention  because of   the large available bandwidth at the mmWave carrier  frequencies (e.g., more than $150$ GHz available bandwidth)\cite{zhang2019precoding}. Thus, mmWave
 communications can potentially provide high data rates. For instance,  IEEE 802.11ad  working on the carrier frequency of $60$ GHz, can support  a maximum data rate of $7$ Gbps\cite{Wang2018}.
   Thanks to the short wavelength of mmWave radio,  large antenna arrays can be packed into mmWave transceivers with limited sizes, thereby resulting highly  directional signals and high  array gains\cite{el2014spatially}. 
  Despite of high data rates,   mmWave communications suffer from severe pathloss and  penetration loss, which limit the coverage of mmWave signals.  For example, the pathloss is about $110$ dB when transmitter-receiver separation distance is $100$ meters \cite{rappaport2017overview}.

Full-duplex (FD) communications, which can support  simultaneous transmission and reception on the same channels, have the potential to double the throughput and reduce latency compared to half-duplex (HD) communications. 
To provide high data rates and improve the coverage of wireless networks,  FD relays have been recently applied in mmWave communications as wireless backhauls\cite{atzeni2017full,han2018precoding,sharma2017dynamic}. Since FD systems  suffer from severe self-interference (SI), SI cancellation (SIC) is one of the main challenges for FD mmWave systems.  
For example, a FD transceiver
at $60$ GHz with a typical transmit power of $14$ dBm, receiver
noise figure of $5$ dB and channel bandwidth of $2.16$ GHz will
require $96$ dB of SIC \cite{dinc2017millimeter}.
Generally, SI can be suppressed by making use of the physical methods, which enhances the propagation loss for the SI signals and maintains a high gain for the desired signals \cite{han2018precoding,satyanarayana2018hybrid,he2017spatiotemporal,dinc2017millimeter}.  
For instance, narrow-beam antennas or beamforming techniques \cite{han2018precoding,satyanarayana2018hybrid} can be used to separate the communication channel and SI channel in directions, and polarization isolation and  antenna spacing\cite{he2017spatiotemporal,dinc2017millimeter} can also be applied.  
Recent measurement in \cite{dinc2017millimeter} shows that almost $80$ dB of SI suppression can be achieved  with  polarization based antennas in  $60$ GHz bands. 
 In addition,  conventional microwave FD systems only consider normal cancellation of the line-of-sight (LOS) SI and  ignore the non-line-of-sight (NLOS) SI. However, in mmWave FD systems, the NLOS SI will be enhanced  due to the high-gain beamforming \cite{zhang2019precoding}. Besides,
 since circuit and hardware complexity scales up  with frequency, conventional full digital processing, which controls both the phases and amplitudes of original signals, becomes very expensive in mmWave systems.  Thus, hybrid beamforming (HBF) that consists of
 digital and analog processing is promising to achieve an optimal trade-off between performance and  complexity. 
 
There have been  few results on the HBF design and SIC for FD mmWave systems \cite{abbas2016full,xiao2017full,han2018precoding,zhang2019precoding}.  Based on the sparsity of mmWave channels, orthogonal matching pursuit (OMP) based HBF algorithm is proposed for  FD mmWave systems \cite{abbas2016full}. However,  SI is not considered in \cite{abbas2016full}, which might significantly affect system performance. In \cite{xiao2017full},  a near-field propagation model is adopted for  line-of-sight (LOS) SI. It is shown that the  SIC performance   can be improved  by increasing the number  of transmitter (TX) or receiver (RX) radio frequency (RF) chains. 
In \cite{han2018precoding},  a combined LOS and  non-line-of-sight (NLOS) SI channel is proposed, and  a decoupled analog-digital (DAD) HBF algorithm is provided to suppress both LOS and NLOS SI. By jointly optimizing the analog and digital precoders, an OMP based  SIC beamforming algorithm for FD mmWave relays is proposed in \cite{zhang2019precoding}. It is shown that the OMP based HBF algorithm can achieve higher spectral efficiency than  DAD HBF algorithm in \cite{han2018precoding}. 
Although the approaches in \cite{zhang2019precoding,han2018precoding}  can perfectly  eliminate SI,  the SIC  of FD mmWave systems is based on the null space of the effective SI channel after  designing optimal  hybrid beamformers,  which will cause a significant degradation in system spectral efficiency. Furthermore, these approaches  can   perform   SIC only when the number of TX-RF chains is greater or equal to the sum of the number of RX-RF chains and the number of transmitting  streams. 
Moreover, in realistic communication systems, since we cannot always obtain perfect channel state information (CSI) through channel estimation, the existing optimization-based FD mmWave HBF approaches cannot provide robust performance in the presence of imperfect CSI. 

Recent development in   machine learning (ML) provides a new way for addressing problems in physical layer communications (e.g., direction-of-arrival estimation\cite{huang2018deep}, analog beam selection \cite{long2018data} and signal detection\cite{samuel2017deep}). ML based techniques have several advantages such as low complexity when solving non-convex problems and the ability  to extrapolate new features from noisy and limited training data \cite{elbir2019cnn}. 
In  \cite{huang2019deep,lin2019beamforming}, precoders are designed  based on  ML techniques, in which a learning network with  multiple fully connected layers is used. 
However, dense  multiple fully connected layers may increase the computational complexity, and these works only optimize the precoder with fixed combiners. In \cite{elbir2019cnn}, a convolutional neural network (CNN) framework is first proposed to jointly optimize the precoder and combiner, in which the network takes the channel matrix as the  input and produces the analog and digital  beamformers  as outputs. To reduce the complexity in training stage in \cite{elbir2019cnn}, an equivalent channel HBF algorithm is proposed   to provide accurate labels for training samples \cite{bao2020deep}. To further reduce the computational complexity,  joint antenna selection and HBF design is studied in \cite{elbir2019joint} based on quantized  CNN with the cost of prediction accuracy degradation.  Though above results can achieve good performance of ML based HBF, all of them consider single-hop scenarios. To the best of our knowledge, the joint  SIC and  HBF design  for FD mmWave relay  systems, being of practical importance, has not been investigated in the context of  ML.   

Motivated by above observations,  we investigate the joint HBF and SIC optimization for FD mmWave relay systems based on  ML techniques. 
 The main contributions of this paper are summarized as follows:

 \begin{itemize}
  \item We decouple the joint SIC and  HBF optimization problem into two sub-problems. We first propose an  alternating direction method of mul-tipliers  (ADMM) based algorithm to jointly eliminate residual SI and optimize unconstrained beamformers. With perfect SIC and unconstrained beamformers, many existing algorithms (e.g., PE-AltMin \cite{yu2016alternating}, GEVD \cite{lin2019hybrid} and methods in \cite{sohrabi2016hybrid})  cannot  be  directly used for HBF design since the  unconstrained beamformers may not be mutually orthogonal.  Thus, we  propose a majorization-minimization (MM)  based algorithm that jointly optimizes the transmitting and receiving HBF. To the best of our knowledge, the ADMM based SIC beamforming and MM based joint transmitting and receiving HBF optimization for FD mmWave systems have not been previously  studied. Unlike  the works  in \cite{zhang2019precoding,han2018precoding}, our proposed approaches can  perform perfect  SIC even if  the  number of  TX  RF  chains  is  smaller  than  the  sum  of  the number of RX RF chains and the number of transmitting  streams.
  Finally, the convergence  and computational complexity of  proposed algorithms are analyzed. 
 \item    Two learning frameworks for HBF design are proposed  (i.e., extreme learning machine  based HBF (ELM-HBF) and CNN based HBF (CNN-HBF)).  We utilize  ELM and CNN  to estimate the precoders and combiners of FD mmWave systems. To support robust HBF performance, noisy channel input data is generated and fed into the learning machine for training. Different from existing optimization based HBF methods, of which the performance strongly relies on the quality of CSI, our learning based approaches can achieve more robust performance since ELM and CNN  are effective at handling the imperfections and  corruptions in the input channel information. 
 
  \item   To the best of our knowledge, HBF design  with ELM  has not been studied before. Also, the performance of ELM-HBF  with different activation functions is tested. Since the optimal weight matrix of hidden layer is derived in a closed-form, the complexity of ELM-HBF is much lower than that of CNN-HBF and   easier for implementation.
  Results show that ELM-HBF can achieve near-optimal performance, which outperforms CNN-HBF and other conventional HBF methods. The training  time of ELM is about  600 times faster than that of  CNN with $64$ transmitting and receiving antennas. While the    conventional methods require an optimization process, our learning based approaches can estimate the beamformers by simply feeding the learning machines with channel matrices. Results also show that,  the  online  prediction  time  of proposed  learning  based  approaches  is  almost  20  times  faster  than the OMP approach. 
\end{itemize}                                

The remainder of this paper is organized as follows. We first present the system model of the FD mmWave relay in Section II. For  SIC and HBF design, we present an  ADMM based SIC beamforming algorithm and an MM based HBF algorithm in Section III. The ELM-HBF and CNN-HBF learning schemes are presented in Section IV. To validate the efficiency of proposed methods,  we  provide  numerical  simulations  in  Section  V. Finally, Section VI concludes the paper. 

\emph{Notations}:  Bold lowercase and uppercase letters denote vectors and matrices, respectively.  $\text{Tr}( {\bf A})$, $\left| {\bf A} \right|$, $\left\| {\bf A} \right\|_\text{F}$, ${\bf A}^*$, ${\bf A}^T$ and ${\bf A}^H$ denote trace, determinant, Frobenius norm, conjugate, transpose and conjugate transpose of matrix ${\bf A}$, respectively. $\otimes$ presents the Kronecker product. $\arg ({\bf a})$ denotes the  argument/phase of vector ${\bf a}$.

\section{system model and problem formulation}\label{systemmodel}

\subsection{System model}
We consider a one-way FD mmWave relay system shown in Fig. 1,  in which the source  node and the destination node are different base stations connected with the FD mmWave relay. We assume that there is no direct link between the  source and destination, which is typical for a mmWave system due to the high pathloss\cite{zhang2019precoding,han2018precoding}. All the  nodes in this system adopt hybrid analog and digital precoding architecture.     
The source node is equipped with  $N_\text{t}$  antennas with  $N_{\text{RFS}}$ RF chains, and transmits $N_\text{s}$ data streams simultaneously. To enable multi-stream transmission, we assume  $N_\text{s} \le N_\text{RFS} \le N_\text{t}$. For the relay and destination nodes, the numbers of antennas and RF chains and the data streams at the relay are defined in the same way, depicted in Fig. 1. 

At the source node, the $N_\text{s} \times 1$ symbol vector, denoted by ${\bf{s}}_\text{S}$ with $\mathbb{E}[{\bf{s}}_\text{S}{\bf{s}}_\text{S}^H]=N_\text{s}^{-1}{\bf I}_{N_\text{s}}$, is firstly precoded through an $N_\text{RFS} \times N_\text{s}$ digital  precoding matrix ${\bf V}_\text{BB}$, and then processed by an $N_\text{t} \times N_\text{RFS} $ analog precoding matrix ${\bf V}_\text{RF}$,  
which is implemented in the analog circuitry using phase
 shifters. Thus, the $N_\text{t} \times 1 $ transmitted signal of the source node is given as
 \begin{equation}\label{tran_signal_s}
 {\bf x}_\text{S} = \sqrt{P_\text{S}}{\bf V}_\text{RF} {\bf V}_\text{BB} {\bf{s}}_\text{S},
 \end{equation}
 where ${P_\text{S}}$ is the transmit power of the source node. The power constraint of the precoding matrices is denoted by $\left\| {\bf V}_\text{RF} {\bf V}_\text{BB} \right\|_\text{F}^2 = N_\text{s}$. Then, the signal received at the relay can be express as 
 \begin{figure*}[t] 
	\vskip 0.2in
	\begin{center}
		\centerline{\includegraphics[width=160mm]{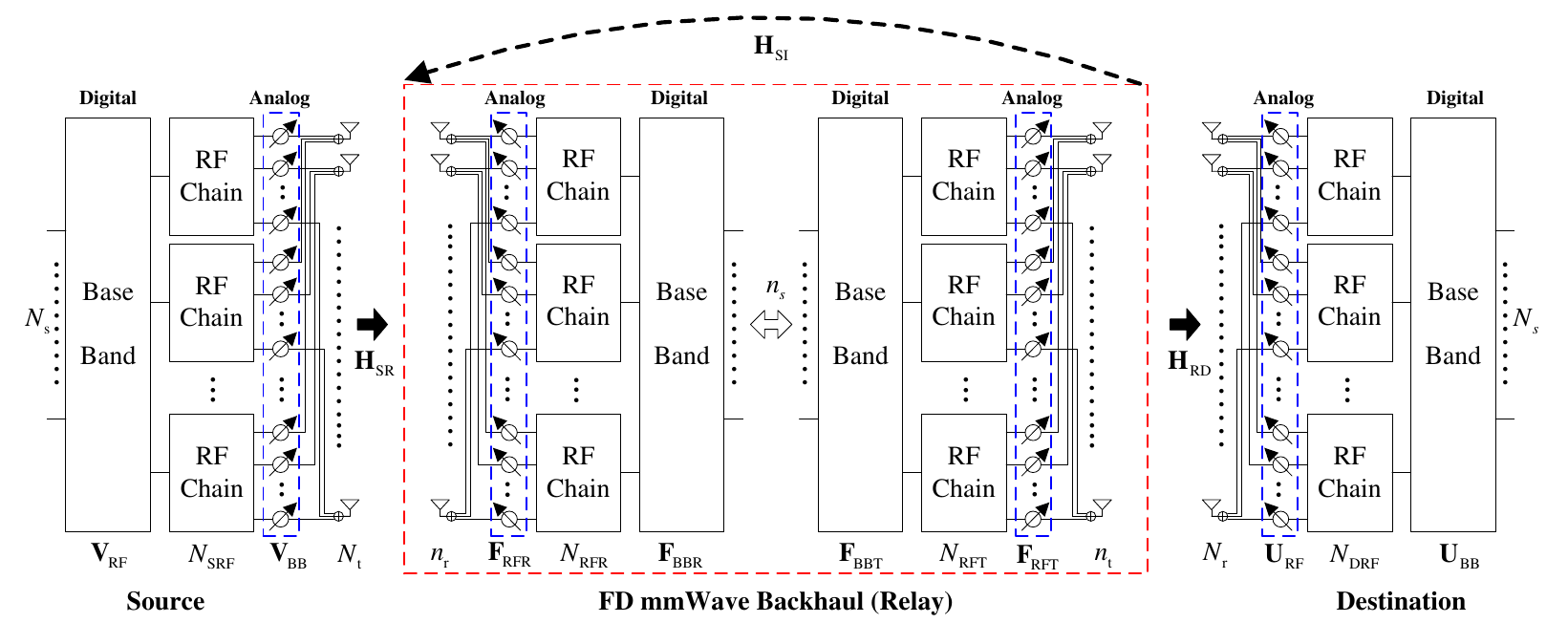}}
		\caption{ FD mmWave relay systems with hybrid analog and digital structure}
		\label{fig_model}
	\end{center}
	\vskip -0.2in
\end{figure*}

   \begin{equation}\label{receive_signal_r}
  {\bf y}_\text{R} = {\bf H}_\text{SR}  {\bf x}_\text{S} + {\bf H}_\text{SI}  {\bf x}_\text{R}+  {\bf n}_\text{R},
  \end{equation}
 where $ {\bf H}_\text{SR} \in \mathbb{C}^{n_\text{r} \times N_\text{t}}$ denotes the  source-to-relay channel matrix, $ {\bf H}_\text{SI} \in \mathbb{C}^{n_\text{r} \times n_\text{t}}$ denotes the SI channel of relay, ${\bf x}_\text{R} \in \mathbb{C}^{n_\text{t} \times1}$ denotes the transmitted signal at the relay, and ${\bf n}_\text{R} \sim \mathcal{CN}(0, \sigma_\text{n}^2 {\bf I}_{n_\text{R}})$  denotes the noise vector at the relay. At the relay, the received signal ${\bf y}_\text{R}$ is  firstly combined with the analog   precoding matrix $ {\bf F}_\text{RFR} \in \mathbb{C}^{n_\text{r} \times N_\text{RFR} }$ and digital precoding matrix $ {\bf F}_\text{BBR} \in \mathbb{C}^{N_\text{RFR}\times n_\text{s}}$. 
 Then, it is precoded by  digital precoding matrix $ {\bf F}_\text{BBT} \in \mathbb{C}^{N_\text{RFT} \times n_\text{s} }$ and analog   precoding matrix $ {\bf F}_\text{RFT} \in \mathbb{C}^{n_\text{t}\times N_\text{RFT} }$. Thus, the transmitted signal at the relay is expressed as 
    \begin{equation}\label{tran_signal_r}
 {\bf x}_\text{R} = \sqrt{P_\text{R}}{\bf F}_\text{RFT} {\bf F}_\text{BBT} {\bf F}_\text{BBR}^H  {\bf F}_\text{RFR}^H {\bf y}_\text{R},
   \end{equation}
 where ${P_\text{R}}$ is the transmit power of the relay. Plugging \eqref{tran_signal_s} and \eqref{receive_signal_r} into \eqref{tran_signal_r}, we obtain 
 \begin{equation}
 \begin{aligned}
{\bf x}_\text{R}） = &{\bf F}_\text{RFT} {\bf F}_\text{BBT}\left( {\bf I}_{n_\text{s}} - \sqrt{P_\text{R}}  {\bf F}_\text{BBR}^H {\bf F}_\text{RFR}^H{\bf H}_\text{SI} {\bf F}_\text{RFT} {\bf F}_\text{BBT} \right)^{-1} \\ 
&\times \sqrt{P_\text{R}}  {\bf F}_\text{BBR}^H {\bf F}_\text{RFR}^H \left( \sqrt{P_\text{S}} {\bf H}_\text{SR} {\bf V}_\text{RF} {\bf V}_\text{BB} {\bf{s}}_\text{S} +   {\bf n}_\text{R}   \right). 
 \end{aligned}
 \end{equation}
At the destination, the received signal is multiplied by the analog combining matrix $ {\bf U}_\text{RF} \in \mathbb{C}^{N_\text{r} \times N_\text{RFD} }$ and digital combining matrix  $ {\bf U}_\text{BB} \in \mathbb{C}^{N_\text{RFD} \times N_\text{s} }$, which is expressed as 
 \begin{equation}\label{receive_destination}
 \begin{split}
{\bf y}_\text{D}）=&
 \sqrt{ P_\text{S}P_\text{R} } {\bf U}^H{\bf H}_\text{RD} {\bf F}_\text{T} {\bf \Xi}_\text{R}^{-1} {\bf F}_\text{R}^H {\bf H}_\text{SR} {\bf V} {\bf{s}}_\text{S} \\
&+\sqrt{P_\text{R} } {\bf U}^H {\bf H}_\text{RD}  {\bf F}_\text{T} {\bf \Xi}_\text{R}^{-1} {\bf F}_\text{R}^H {\bf n}_\text{R} 
  + {\bf U}^H {\bf n}_\text{D}, 
 \end{split}
 \end{equation} 
where ${\bf V}={\bf V}_\text{RF}{\bf V}_\text{BB}$, ${\bf U} ={\bf U}_\text{RF}{\bf U}_\text{BB}$, ${\bf F}_\text{T} ={\bf F}_\text{RFT} {\bf F}_\text{BBT}$,  ${\bf F}_\text{R} ={\bf F}_\text{RFR} {\bf F}_\text{BBR}$, 
and ${\bf \Xi }_\text{R}=  {\bf I}_{n_\text{s}} - \sqrt{P_\text{R}}  {\bf F}_\text{R}^{H}  {\bf H}_\text{SI} {\bf F}_\text{T} $. $ {\bf H}_\text{RD} \in \mathbb{C}^{N_\text{r} \times n_\text{t}}$ denotes the  relay-to-destination channel matrix and 
${\bf n}_\text{D} \sim \mathcal{CN}(0, \sigma_\text{n}^2 {\bf I}_{N_\text{r}})$  denotes the noise vector at  destination.

From the above, the spectral efficiency of the system is given by 
\begin{equation}\label{SE}
 \begin{split}
R =&\log_2  \left| {\bf I}_{N_\text{s}} +\frac{P_\text{S}P_\text{R}}{N_\text{s}}{\bf \Sigma}^{-1} \left( {\bf U}^{H} {\bf H}_\text{RD} {\bf F}_\text{T} {\bf \Xi}_\text{R}^{-1} {\bf F}_\text{R}^{H} {\bf H}_\text{SR} {\bf V} \right) \right.\\
& \left. \times \left( {\bf U}^{H} {\bf H}_\text{RD} {\bf F}_\text{T} {\bf \Xi}_\text{R}^{-1} {\bf F}_\text{R}^{H} {\bf H}_\text{SR} {\bf V} \right)^{H}  \right|, 
 \end{split}
\end{equation}
where ${\bf \Sigma} = \sigma_\text{n}^2 \big[ {P_\text{R} } \left({\bf U}^{H} {\bf H}_\text{RD}  {\bf F}_\text{T} {\bf \Xi}_\text{R}^{-1} {\bf F}_\text{R}^{H} \right) \left({\bf U}^{H} {\bf H}_\text{RD}  {\bf F}_\text{T} {\bf \Xi}_\text{R}^{-1} {\bf F}_\text{R}^{H} \right)^{H} +{\bf U}^{H} {\bf U} \big]$ is the covariance matrix of the noise term in \eqref{receive_destination}.

\subsection{Channel model}
For the desired link channels (i.e., ${\bf H}_\text{SR}$  and ${\bf H}_\text{RD}$), we  assume that sufficient far-field conditions have been met, and employ 
the extended Saleh-Valenzuela mmWave channel model \cite{Ming2017xiao,zhang2019precoding} to characterize the limited scattering features of mmWave channels. This model is described  as the sum of the contributions from $N_\text{c}$ scattering clusters, each of which contributes  $N_\text{p}$ propagation paths. This model expresses the mmWave channel as 
\begin{equation}\label{channelmodel_SD}
{\bf H} =\sqrt{\frac{N_\text{R}N_\text{T}}{ N_\text{c} N_\text{p}}}\sum_{k=1}^{N_\text{c}} \sum_{l=1}^{N_\text{p}} \alpha_{k,l} {\bf{a}}(\theta_{k,l}^{\text{r}}) {\bf{a}}(\theta_{k,l}^{\text{t}})^{H},     
\end{equation}
 where $N_\text{T}$ denotes the number of  transmit antennas, $N_\text{R}$ denotes the number of  receive antennas, $\alpha_{k,l}$ denotes the complex gain of the $l$-th ray in the $k$-th propagation cluster.  The functions  ${\bf{a}}(\theta_{k,l}^{\text{r}})$ and $ {\bf{a}}(\theta_{k,l}^{\text{t}})^H$  respectively represent the normalized receive and  transmit  array response vectors, where $\theta_{k,l}^{\text{r}}$ and $\theta_{k,l}^{\text{t}}$ are the azimuth angles of arrival and departure, respectively. The array response can be expressed as    
\begin{equation}\label{ULA}
{\bf{a}}(\theta ) = \frac{1}{{\sqrt {{N}} }}{\left[1,{e^{ - j\frac{{2\pi d}}{\lambda }\sin(\theta) }},...,{e^{ - j({N} - 1)\frac{{2\pi d}}{\lambda }\sin(\theta) }}\right]^T},
\end{equation}
where $d$ is the antenna spacing and $\lambda$ is the carrier wave-length.  

\begin{figure*}[b]
 	\hrulefill 
 \setcounter{equation}{9}
\begin{equation}\label{distance_d}
d_{m,n}=\sqrt{(a_0+(m-1)d)^2 + (b_0+(n-1)d)^2 -2(a_0 +(m-1)d)(b_0+(n-1)d)\cos(\phi) }, 
\end{equation}  
\hrulefill 
 \setcounter{equation}{11}
\begin{equation}\label{Problem_Formulation}
\begin{split}
  \mathop {\max }\limits_{\scriptstyle {\bf V}_\text{RF}, {\bf U}_\text{RF},{\bf F}_\text{RFT},{\bf F}_\text{RFR} \hfill\atop
\scriptstyle {\bf V}_\text{BB}, {\bf U}_\text{BB},{\bf F}_\text{BBT},{\bf F}_\text{BBR}   \hfill}   
 &R = \log_2  \left| {\bf I}_{N_\text{s}} +\frac{P_\text{S}P_\text{R}}{N_\text{s}}\Sigma^{-1} \left[ {\bf U}^{H} {\bf H}_\text{RD} {\bf F}_\text{T}   {\bf F}_\text{R}^{H} {\bf H}_\text{SR} {\bf V} \right]  \left[ {\bf U}^{H} {\bf H}_\text{RD} {\bf F}_\text{T}   {\bf F}_\text{R}^{H} {\bf H}_\text{SR} {\bf V} \right]^{H}  \right|\\
  \text{s.t.}  \qquad~~ &{\bf V}_\text{RF}\in \mathcal{W}_\text{RF}, {\bf U}_\text{RF}\in \mathcal{G}_\text{RF}, {\bf F}_\text{RFR}\in \mathcal{F}_\text{RFR},  {\bf F}_\text{RFT}\in \mathcal{F}_\text{RFT},\\
&  \left\| {\bf F}_\text{RFT} {\bf F}_\text{BBT} \right\|_\text{F}^2 = N_\text{s},\\
& \left\| {\bf V}_\text{RF} {\bf V}_\text{BB} \right\|_\text{F}^2 = N_\text{s},\\
& {\bf F}_\text{R}^{H} {\bf H}_\text{SI} {\bf F}_\text{T} =0,
\end{split}
\end{equation}
\end{figure*}

Based on the widely used mmWave SI channel in \cite{zhang2019precoding,han2018precoding,satyanarayana2018hybrid}, the residual SI consists of two parts: 
 LOS SI and  NLOS SI.  With the high beam gain of mmWave signals,
the NLOS SI results from refection from nearby obstacles. We can utilize the far-filed clustered mmWave channel in  \eqref{channelmodel_SD} to model the NLOS SI channel. 
By contrast, in the FD scenario, since the transmitter and the local receiver are closely placed, LOS SI channels are near-field channels and clustered mmWave channels  cannot be used.
Therefore, a  more realistic   SI channel model, which
is the spherical wave propagation model,  is considered for the near-field LOS channel matrix. 
 According to \cite{zhang2019precoding,satyanarayana2018hybrid}, the LOS SI channel coefficient of the $m$-th row and $n$-th column entry is given by 
 \setcounter{equation}{8}
\begin{equation}
\left[{\bf H}_{\text{LOS}}\right]_{m,n} =\frac{\rho}{d_{m,n}} \exp\left(-j\frac{2\pi}{\lambda}d_{m,n}\right),
\end{equation}
where $\rho$ is a normalization constant such that $\mathbb{E} [  \left\|{\bf H}_{\text{LOS}} \right\|_\text{F}^2  ] = n_\text{t} n_\text{r} $ and $d_{m,n}$ is the distance from the $m$-th element of the transmit array to the $n$-th element of the receive array, expressed in \eqref{distance_d}, with $\phi$ corresponding to the angle between the two antenna arrays, and  $a_0$ and $b_0$ denoting the  initial distances of the antennas to a common reference point. Finally, the FD mmWave SI channel is constructed as  
\setcounter{equation}{10}
\begin{equation}
{\bf H}_{\text{SI}}=\kappa_{\text{LOS}} {\bf H}_{\text{LOS}} + \kappa_{\text{NLOS}} {\bf H}_{\text{NLOS}},
\end{equation}
where $\kappa_{\text{LOS}}$ and $\kappa_{\text{NLOS}}$ denote the intensity coefficients of LOS and NLOS components, respectively. 

\subsection{Problem formulation}
The main objective of this work is to maximize the system spectral efficiency by jointly designing  beamforming matrices (i.e., ${\bf V}_\text{RF}$, ${\bf V}_\text{BB}$, ${\bf F}_\text{RFT} $, ${\bf F}_\text{BBT}$,  ${\bf F}_\text{RFR}$, ${\bf F}_\text{BBR}$, ${\bf U}_\text{RF}$  and ${\bf U}_\text{BB}$)  and SIC in the presence of noisy  CSI. The  HBF design  problem of maximizing the  system spectral efficiency  is given by \eqref{Problem_Formulation}, where $\mathcal{V}_\text{RF}$, $\mathcal{U}_\text{RF}$, $\mathcal{F}_\text{RFR}$ and $\mathcal{F}_\text{RFT}$ are the feasible sets of the analog precoders included by unit modulus constraints. 
Due to the non-convex constraints of the analog precoders and  ${\bf F}_\text{R}^{H} {\bf H}_\text{SI} {\bf F}_\text{T} =0$, it is in general intractable to solve the optimization problem. 
Meanwhile,  the HBF design with conventional optimization-based approaches are not robust  in the presence of noisy CSI. To solve these problems, we first develop   algorithms that maximize the spectral efficiency with joint HBF and SIC  design. Then, learning machines (e.g., ELM and CNN) are adopted such that the hybrid beamformers are predicted by feeding the machines with noisy CSI.  

\section{FD mmWave beamforming design}\label{FD_mmWave_beamforming_design} 
In order to train the ELM and CNN learning machines, we first need to solve the optimization problem in \eqref{Problem_Formulation} to provide accurate labels of the training data samples.  Generally, to make the problem tractable and reduce the communication overhead, the beamformers can be designed individually at different nodes \cite{zhang2019precoding,han2018precoding}. The main challenge  is at the relay node since  the beamformers of the transmitting part and receiving part should be jointly designed, and  SIC needs to be guaranteed.  
Consequently, we first propose efficient algorithms to design the hybrid beamformers at the relay, and then HBF design at the source and destination can be obtained following similar approaches at relay. 

\subsection{SIC and HBF algorithm design}
According to the method studied in \cite{el2014spatially,lee2014af,tsinos2018hybrid}, the HBF design problem at the relay can be transferred  into minimizing the Frobenius norm of the difference between the optimal fully digital beamformer and hybrid beamformers as 
 \setcounter{equation}{12}
\begin{equation}\label{Problem_relay}
\begin{split}
(\text{P}1):~ \mathop {\min }\limits_{{\bf F}_\text{RFT}, {\bf F}_\text{BBT},{\bf F}_\text{RFR},{\bf F}_\text{BBR}}& \left\| {\bf F}_\text{opt} -  {\bf F}_\text{RFT} {\bf F}_\text{BBT} {\bf F}_\text{BBR}^{H} {\bf F}_\text{RFR}^{H} \right\|_\text{F}^2 \\
   \text{s.t.} ~~ \qquad   & {\bf F}_\text{RFR}\in \mathcal{F}_\text{RFR},  {\bf F}_\text{RFT}\in \mathcal{F}_\text{RFT},\\
   &\left\| {\bf F}_\text{RFT} {\bf F}_\text{BBT} \right\|_\text{F}^2 = n_\text{s},\\
 & {\bf F}_\text{BBR}^{H} {\bf F}_\text{RFR}^{H} {\bf H}_\text{SI} {\bf F}_\text{RFT} {\bf F}_\text{BBT} =0,
\end{split}
\end{equation}  
where  ${\bf F}_\text{opt} = {\bf F}_\text{TSVD} {\bf F}_\text{RSVD}^H $, ${\bf F}_\text{TSVD} $ and ${\bf F}_\text{RSVD}$ are formed by the  right-singular vectors of ${\bf H}_\text{RD} $  and  left-singular vectors of ${\bf H}_\text{SR} $, respectively. 

Due to the non-convex constraints, it is still hard to find a solution of problem (P1) that guarantees both SIC and  optimal hybrid beamformers. Moreover,  according to the results in \cite{zhang2019precoding, han2018precoding}, the spectral efficiency of FD mmWave systems degrades significantly if residual SI cannot be efficiently suppressed. Thus, we will solve this problem following two steps. In the first step, we mainly focus on perfectly eliminating  SI (i.e., ${\bf F}_\text{BBR}^{H} {\bf F}_\text{RFR}^{H} {\bf H}_\text{SI} {\bf F}_\text{RFT} {\bf F}_\text{BBT} =0$) by designing unconstrained beamformers for both transmitting and receiving parts. Then, the problem (P1) can be rewritten as  
\begin{equation}\label{Problem_relay_step1}
\begin{split}
(\text{P}2):~\mathop {\min }\limits_{{\bf F}_\text{T}, {\bf F}_\text{R} }~ &\left\| {\bf F}_\text{opt} -    {\bf F}_\text{T} {\bf F}_\text{R}^{H} \right\|_\text{F}^2 \\
   \quad \text{s.t.}  ~~  &{\bf F}_\text{R}^{H} {\bf H}_\text{SI} {\bf F}_\text{T} =0. \\
\end{split}
\end{equation} 
After solving (P2), we obtain the unconstrained beamformers that can ensure perfect SIC. In the second step, hybrid beamformers of the transmitting part and receiving part are jointly designed by minimizing the Frobenius norm of the difference between the unconstrained beamformers and hybrid beamformers. The problem is formulated as 
\begin{equation}\label{Problem_relay_step2}
\begin{split}
(\text{P}3):~ \mathop {\min }\limits_{{\bf F}_\text{RFT}, {\bf F}_\text{BBT},{\bf F}_\text{RFR},{\bf F}_\text{BBR}} ~&\left\| \hat{\bf F}_\text{opt} -  {\bf F}_\text{RFT} {\bf F}_\text{BBT} {\bf F}_\text{BBR}^{H} {\bf F}_\text{RFR}^{H} \right\|_\text{F}^2 \\
  \qquad \text{s.t.} ~~~~ \quad\quad    &{\bf F}_\text{RFR}\in \mathcal{F}_\text{RFR},  {\bf F}_\text{RFT}\in \mathcal{F}_\text{RFT},\\
  \qquad \quad \qquad   & \left\| {\bf F}_\text{RFT} {\bf F}_\text{BBT} \right\|_\text{F}^2 = n_\text{s},\\
\end{split}
\end{equation} 
where $\hat{\bf F}_\text{opt} = \hat{\bf F}_\text{T} \hat{\bf F}_\text{R}^{H}$ with $ \hat{\bf F}_\text{T}$ and $ \hat{\bf F}_\text{R}$ corresponding to the solutions of problem (P2). Though   problems (P2) and (P3) are simpler than the original problem (P1), both problems are still non-convex problems and efficient approaches should be proposed to solve them. 


Let us start with   problem (P2). It is obvious that this problem is convex when  one of the variables is fixed. This property enables ADMM utilization, and consequently, the augmented Lagrangian function of  (P2) is given by 
\begin{equation}
\mathcal{L}({\bf F}_\text{T},{\bf F}_\text{R}, {\bf Z})= \left\| {\bf F}_\text{opt} -    {\bf F}_\text{T} {\bf F}_\text{R}^{H} \right\|_\text{F}^2 + \varrho \left\| {\bf F}_\text{R}^{H} {\bf H}_\text{SI} {\bf F}_\text{T} + \frac{1}{\varrho} {\bf Z}  \right\|_\text{F}^2, 
\end{equation}
where ${\bf Z} \in \mathbb{C}^{n_\text{s} \times n_\text{s}}$ is the Lagrange multiplier matrix and $\varrho$ is the ADMM  step-size. According to ADMM approaches\cite{boyd}, the solution to (P2) can be obtained iteratively, where in iteration $i+1$, the variables and multiplier are updated as follows  
\begin{subequations}
	\begin{align}		
	{\bf F}_\text{T}^{(i+1)}:=& \arg \mathop{\min}\limits_{{\bf F}_\text{T}} \mathcal{L}\left({\bf F}_\text{T} ,{\bf F}_\text{R}^{(i)},{\bf Z}^{(i)}\right); \label{Problem_FT}\\
	{\bf F}_\text{R}^{(i+1)}:=& \arg \mathop{\min}\limits_{{\bf F}_\text{R}} \mathcal{L}\left({\bf F}_\text{T}^{(i+1)},{\bf F}_\text{R} ,{\bf Z}^{(i)} \right); \label{Problem_FR}\\
	{\bf Z}^{(i+1)} :=& {\bf Z}^{(i)} +\varrho {\bf F}_\text{R}^{H{(i+1)}} {\bf H}_\text{SI} {\bf F}_\text{T}^{(i+1)}.\label{eq_Z}
	\end{align}	
\end{subequations}

By solving the problems in \eqref{Problem_FT} and \eqref{Problem_FR},  we have the following theorem.

\begin{theorem}
The closed-form expressions for ${\bf F}_{\rm{T}}^{(i+1)}$ and  ${\bf F}_{\rm{R}}^{(i+1)}$ are receptively given by 
\begin{equation}\label{Problem_FT_solve}
\begin{split}
\text{vec}\left({\bf F}_{\rm{T}}^{(i+1)}\right)=&\left[ {\bf I}_{n_{\rm{s}}} \otimes \left( \varrho {\bf H}_{\rm{SI}}^{H} {\bf F}_{\rm{R}}^{(i)} {\bf F}_{\rm{R}}^{H(i)}  {\bf H}_{\rm{SI}} \right) +\left( {\bf F}_{\rm{R}}^{H(i)} {\bf F}_{\rm{R}}^{(i)} \right)  
  \otimes {\bf I}_{n_{\rm{t}}}        \right]^{-1} \\
&\times\text{vec}\left({\bf F}_{\rm{opt}} {\bf F}_{\rm{R}}^{(i)} - {\bf H}_{\rm{SI}}^{H} {\bf F}_{\rm{R}}^{(i)}  {\bf Z}^{(i)}  \right), \\
\end{split}
\end{equation}
and
\begin{equation}\label{Problem_FR_solve}
\begin{split}
&\text{vec}\left({\bf F}_{\rm{R}}^{(i+1)}\right)\\&=\left[ \left( \varrho {\bf H}_{\rm{SI}}  {\bf F}_{\rm{T}}^{(i+1)} {\bf F}_{\rm{T}}^{H(i+1)} {\bf H}_{\rm{SI}}^{H}\right)    \otimes {\bf I}_{n_{\rm{s}}}    +  {\bf I}_{n_{\rm{r}}}\otimes \left( {\bf F}_{\rm{T}}^{H(i+1)} {\bf F}_{\rm{T}}^{(i+1)}\right)  \right]^{-1}\\&~~~\times \text{vec}\left( {\bf F}_{\rm{T}}^{H(i+1)}  {\bf F}_{\rm{opt}} - {\bf Z}^{(i)} {\bf F}_{\rm{T}}^{H(i+1)}  {\bf H}_{\rm{SI}}^{H}  \right).  \\
\end{split}
\end{equation} 
\end{theorem}
\begin{proof}
We first solve problem \eqref{Problem_FT}. By taking the derivative of $\mathcal{L}\big({\bf F}_\text{T},{\bf F}_\text{R}^{(i)},{\bf Z}^{(i)}\big)$ over ${\bf F}_\text{T}$ and letting it to be zero, we have ${\bf F}_\text{T}{\bf F}_\text{R}^{H(i)}{\bf F}_\text{R}^{(i)}+ \varrho {\bf H}_\text{SI}^{H} {\bf F}_\text{R}^{(i)} {\bf F}_\text{R}^{H(i)} {\bf H}_\text{SI}{\bf F}_\text{T} ={\bf F}_\text{opt}{\bf F}_\text{R}^{(i)}-{\bf H}_\text{SI}^{H} {\bf F}_\text{R}^{(i)} {\bf Z}^{(i)}$. Then,  utilizing the  vectorization property, the result in \eqref{Problem_FT_solve} can be obtained. For the problem \eqref{Problem_FR}, we can use a similar approach of the problem \eqref{Problem_FT} and obtain ${\bf F}_\text{R}^{(i+1)}$ in \eqref{Problem_FR_solve}.
\end{proof}

The above alternating procedure is initialized by setting the entries of  matrices  ${\bf F}_\text{R}^{(0)}$ to random values, and multiplier ${\bf Z}^{(0)}$ to zeros. 
After obtaining the updated variables in each steps, we  summarize the ADMM-based  beamforming design and  SIC approach in Algorithm \ref{ALG1}.

\begin{algorithm}[t]\label{ALG1}
	\caption{ADMM-based beamforming and SI cancellation algorithm} 
	\begin{algorithmic}[1]
      	\STATE \textbf{Input}: ${\bf F}_\text{opt}$, ${\bf H}_\text{SI}$, $\varrho$;
      	\STATE \textbf{Output}: ${\bf F}_\text{T}$, ${\bf F}_\text{R}$;
		\STATE \textbf{Initialize}: ${\bf F}_\text{R}^{(0)}$, ${\bf Z}^{(0)}$ and  $i=0$;  
		\REPEAT
		\STATE \text{Update} ${\bf F}_\text{T}^{(i+1)}$ using \eqref{Problem_FT_solve};
		\STATE  \text{Update} ${\bf F}_\text{R}^{(i+1)}$ using \eqref{Problem_FR_solve};
		\STATE  \text{Update} ${\bf Z}^{(i+1)}$ using \eqref{eq_Z};
		\STATE $i\leftarrow i+1$;
		\UNTIL{the stopping criteria is  met.}
	\end{algorithmic} 
\end{algorithm}

After obtaining the solutions of problem (P2), we then turn to solve problem (P3). The main challenges for this problem are the constant modulus non-convex constraints and the joint optimization of transmitting and receiving hybrid beamformers.  Furthermore, since the  beamformers derived by Algorithm \ref{ALG1} for SIC may not be mutually orthogonal, many existing approaches  (e.g., PE-AltMin \cite{yu2016alternating}, GEVD \cite{lin2019hybrid} and methods in \cite{sohrabi2016hybrid}) cannot be  directly used for this case. To make this problem tractable and deal with the non-convex constraints, we utilize majorization-minimization (MM) methods \cite{sun2016majorization,wu2017transmit,arora2019hybrid}.
In stead of minimizing the original objective function directly, the MM procedure consists of two steps. In the first majorization step, we find an easy implemented surrogate function that should be a tight upper bound of the original objective function and exist closed-from minimizers. Then in the minimization step, we minimize the surrogate function with closed-from minimizers. 
To achieve a fast convergence rate, a surrogate function that tries to follow the shape of the objective function is preferable\cite{sun2016majorization}. 

 To jointly design the hybrid beamformers at transmitting and receiving parts with MM methods, we will solve problem (P3) based on the alternating minimization framework.  
We first solve problem (P3) for the transmitting precoder ${\bf F}_\text{RFT}$ by fixing ${\bf F}_\text{RFR}$, ${\bf F}_\text{BBR}$ and ${\bf F}_\text{BBT}$.  Problem (P3) can be rewritten as 
 \begin{equation}\label{Problem_relay_step2_1}
 \begin{split}
 (\text{P}4):~ &\mathop {\min }\limits_{{\bf F}_\text{RFT}} ~     \left\| \hat{\bf F}_\text{opt} -  {\bf F}_\text{RFT} {\bf Y}_\text{T}   \right\|_\text{F}^2 \\
  &  ~ \text{s.t.} ~~~~      {\bf F}_\text{RFT}\in \mathcal{F}_\text{RFT},\\
 \end{split}
 \end{equation} 
 where $ {\bf Y}_\text{T}={\bf F}_\text{BBT} {\bf F}_\text{BBR}^H {\bf F}_\text{RFR}^H$. Then, we rewrite the objective function of problem (P4) as  
 \begin{equation}\label{J_T}
 \begin{split}
 J({\bf F}_\text{RFT}; {\bf Y}_\text{T} )&=  \text{Tr}\left( \hat{\bf F}_\text{opt} \hat{\bf F}_\text{opt}^H\right)+ \text{Tr}\left( {\bf F}_\text{RFT}^H  {\bf F}_\text{RFT} {\bf Y}_\text{T} {\bf Y}_\text{T}^H  \right)\\
 &\quad - \text{Tr}\left( \hat{\bf F}_\text{opt} {\bf Y}_\text{T}^H {\bf F}_\text{RFT}^H \right)- \text{Tr}\left( {\bf F}_\text{RFT} \left(\hat{\bf F}_\text{opt} {\bf Y}_\text{T}^H\right)^H \right)\\
& \mathop  = \limits^{(a)} \text{Tr}\left( \hat{\bf F}_\text{opt} \hat{\bf F}_\text{opt}^H \right)+{\bf f}_\text{RFT}^H  {\bf Q}_\text{T} {\bf f}_\text{RFT} -2\text{Re}\left({\bf f}_\text{RFT}^H  {\bf e}_\text{T}\right),  
 \end{split}
 \end{equation} 
where ${\bf f}_\text{RFT}= \text{vec}({\bf F}_\text{RFT})$, ${\bf E}_\text{T}=\hat{\bf F}_\text{opt} {\bf Y}_\text{T}^H$, ${\bf e}_\text{T}=\text{vec}({\bf E}_\text{T})$, ${\bf Q}_\text{T}=({\bf Y}_\text{T}{\bf Y}_\text{T}^H  )^T \otimes {\bf I}_{n_\text{t}} $, and  (a) follows from the identity $\text{Tr}({\bf A}{\bf B}{\bf C}{\bf D})=\text{vec}({\bf A}^T)^T ( {\bf D}^T\otimes{\bf B}  )\text{vec}({\bf C})$.

To solve  problem (P4) with  MM methods, we should find a majorizer of  $J({\bf F}_\text{RFT}; {\bf Y}_\text{T} )$ according to the following lemma. 

\begin{lemma}\label{Quad} 
Let ${\bf Q} \in \mathbb{C}^{N\times N}$ and ${\bf S} \in \mathbb{C}^{N\times N}$ be two   Hermitian matrices satisfying
${\bf S} \ge {\bf Q}$. 
Then the quadratic function  ${\bf a}^H {\bf Q} {\bf a}$,  is majorized by 
${\bf a}^H {\bf S} {\bf a} + 2\text{Re}( {\bf a}^H ( {\bf Q} -{\bf S} ) {\bf a}_i )+{\bf a}_i^H ( {\bf S}-{\bf Q} ) {\bf a}_i $ at  point ${\bf a}_i\in \mathbb{C}^{N}$.
\end{lemma}
\begin{proof}
The proof can be found in \cite{wu2017transmit}.
\end{proof}
According to  Lemma \ref{Quad}, we can obtain a valid majorizer of  $J({\bf F}_\text{RFT}; {\bf Y}_\text{T} )$ at point ${\bf F}_\text{RFT}^{(i)} \in \mathcal{F}_\text{RFT}$ given by 
 \begin{equation}\label{J_T_major}
 \begin{split}
 &\bar J\left({\bf F}_\text{RFT}; {\bf Y}_\text{T},{\bf F}_\text{RFT}^{(i)} \right)\\
  = & \text{Tr}\left( \hat{\bf F}_\text{opt} \hat{\bf F}_\text{opt}^H\right)+\lambda_\text{T} {\bf f}_\text{RFT}^H  {\bf f}_\text{RFT} + 2\text{Re}\left( {\bf f}_\text{RFT}^H \left( {\bf Q}_\text{T} -\lambda_\text{T}{\bf I}\right) {\bf f}_\text{RFT}^{(i)} \right)\\
  & +{\bf f}_\text{RFT}^{H(i)}\left( \lambda_\text{T}{\bf I} -{\bf Q}_\text{T}  \right){\bf f}_\text{RFT}^{(i)} -2\text{Re}\left({\bf f}_\text{RFT}^H  {\bf e}_\text{T}\right)\\
  = & 2\text{Re}\left( {\bf f}_\text{RFT}^H \left( \left( {\bf Q}_\text{T} -\lambda_\text{T}{\bf I}\right){\bf f}_\text{RFT}^{(i)} -  {\bf e}_\text{T}   \right) \right) +C_\text{T},
 \end{split}
 \end{equation} 
where ${\bf F}_\text{RFT}^{(i)}$ is the iterate available at $i$-th iteration, $\lambda_\text{T}$ denotes the maximum eigenvalue of  ${\bf Q}_\text{T}$, and the constant term $C_\text{T}=\text{Tr}( \hat{\bf F}_\text{opt} \hat{\bf F}_\text{opt}^H)+\lambda_\text{T} {\bf f}_\text{RFT}^H  {\bf f}_\text{RFT}+{\bf f}_\text{RFT}^{H(i)}( \lambda_\text{T}{\bf I} -{\bf Q}_\text{T}  ){\bf f}_\text{RFT}^{(i)}$.  Thus, we can guarantee $\lambda_\text{T}{\bf I} \ge {\bf Q}_\text{T}$. 
Then,  utilizing the majorizer in \eqref{J_T_major}, the solution of problem (P4) can be obtained by  iteratively solving the following problem 
  \begin{equation}\label{Problem_relay_step2_12}
  \begin{split}
  (\text{P}5):~ &{\bf F}_\text{RFT}^{(i+1)}=\arg \mathop {\min }\limits_{{\bf F}_\text{RFT}} ~ \bar J\left({\bf F}_\text{RFT}; {\bf Y}_\text{T},{\bf F}_\text{RFT}^{(i)} \right) \\
   &  ~~ \text{s.t.} ~~~   {\bf F}_\text{RFT}\in \mathcal{F}_\text{RFT}.\\
  \end{split}
  \end{equation} 
The close-form solution of problem (P5) is given by 
\begin{equation}\label{frft}
{\bf f}_\text{RFT}^{(i+1)}=-\exp\left(j\arg\left( \left( {\bf Q}_\text{T} -\lambda_\text{T}{\bf I}\right){\bf f}_\text{RFT}^{(i)} -  {\bf e}_\text{T}   \right)\right).
\end{equation}

Similarly, we can solve problem (P3) for the receiving combiner  ${\bf F}_\text{RFR}$ by fixing ${\bf F}_\text{RFT}$, ${\bf F}_\text{BBR}$ and ${\bf F}_\text{BBT}$.  Then, problem (P3) can be rewritten as 
 \begin{equation}\label{Problem_relay_step2_2}
 \begin{split}
 (\text{P}6):~ &\mathop {\min }\limits_{{\bf F}_\text{RFR}} ~ \text{Tr}\left( \hat{\bf F}_\text{opt} \hat{\bf F}_\text{opt}^H\right)+{\bf f}_\text{RFR}^H  {\bf Q}_\text{R} {\bf f}_\text{RFR} -2\text{Re}\left({\bf f}_\text{RFR}^H  {\bf e}_\text{R}\right) \\
  &  ~~ \text{s.t.} ~~~     {\bf F}_\text{RFR}\in \mathcal{F}_\text{RFR},\\
 \end{split}
 \end{equation} 
 where 
  ${\bf Q}_\text{R}=({\bf Y}_\text{R}^H {\bf Y}_\text{R}   )^T \otimes {\bf I}_{n_\text{r}}$, ${\bf Y}_\text{R} = {\bf F}_\text{RFT} {\bf F}_\text{BBT} {\bf F}_\text{BBR}^H$,  ${\bf f}_\text{RFR}= \text{vec}({\bf F}_\text{RFR})$, ${\bf E}_\text{R}=\hat{\bf F}_\text{opt}^H {\bf Y}_\text{R} $, and  ${\bf e}_\text{R}=\text{vec}({\bf E}_\text{R})$. According to  Lemma \ref{Quad}, we can obtain a valid majorizer of the objective function in \eqref{Problem_relay_step2_2} at point ${\bf F}_\text{RFR}^{(i)} \in \mathcal{F}_\text{RFR}$, which is given by 
 \begin{equation}\label{J_R_major}
 \begin{split}
\tilde  J\left({\bf F}_\text{RFR}; {\bf Y}_\text{R}, {\bf F}_\text{RFR}^{(i)} \right) =  2\text{Re}\left( {\bf f}_\text{RFR}^H \left( \left( {\bf Q}_\text{R} -\lambda_\text{R}{\bf I}\right){\bf f}_\text{RFR}^{(i)} -  {\bf e}_\text{R}   \right) \right) +C_\text{R},
 \end{split}
 \end{equation} 
where $\lambda_\text{R}$ denotes the maximum eigenvalue of  ${\bf Q}_\text{R}$ and the constant term $C_\text{R}=\text{Tr}( \hat{\bf F}_\text{opt} \hat{\bf F}_\text{opt}^H)+\lambda_\text{R} {\bf f}_\text{RFR}^H  {\bf f}_\text{RFR}+{\bf f}_\text{RFR}^{H(i)}( \lambda_\text{R}{\bf I} -{\bf Q}_\text{R}  ){\bf f}_\text{RFR}^{(i)}$.  
Following  a  similar procedure of solving problem (P4), the solution of problem (P6) can be obtained by  iteratively  updating  $ {\bf F}_\text{RFR}$ according to the following close-form expression  
\begin{equation}\label{frfr}
{\bf f}_\text{RFR}^{(i+1)}=-\exp\left(j\arg\left( \left( {\bf Q}_\text{R} -\lambda_\text{R}{\bf I}\right){\bf f}_\text{RFT}^{(i)} -  {\bf e}_\text{R}   \right)\right).
\end{equation}

Then, we turn to design digital beamformers (i.e., ${\bf F}_\text{BBR}$ and ${\bf F}_\text{BBT}$) with fixed analog beamformers (i.e., ${\bf F}_\text{RFR}$ and ${\bf F}_\text{RFT}$). By fixing ${\bf F}_\text{BBT}$, ${\bf F}_\text{RFR}$ and ${\bf F}_\text{RFT}$, a  globally optimal  solution of problem (P3) is given by
 \begin{equation}\label{fbbr}
{\bf F}_\text{BBR} = {\bf F}_\text{RFR}^{-1} \hat{\bf F}_\text{opt}^H \left( {\bf F}_\text{BBT}^H {\bf F}_\text{RFT}^H  \right)^{-1}. 
 \end{equation}
Similarly, by fixing ${\bf F}_\text{BBR}$, ${\bf F}_\text{RFR}$ and ${\bf F}_\text{RFT}$, a solution of problem (P3) without considering the power constraint in \eqref{Problem_relay_step2} is given by
 \begin{equation}\label{fbbt}
{\bf F}_\text{BBT} = {\bf F}_\text{RFT}^{-1} \hat{\bf F}_\text{opt} \left( {\bf F}_\text{BBR}^H {\bf F}_\text{RFR}^H  \right)^{-1}. 
 \end{equation} 
Then, to satisfy the power constraint in problem (P3), we can normalize ${\bf F}_\text{BBT}$ by a factor of $\frac{\sqrt{n_\text{s}}}{    \left\| {\bf F}_\text{RFT} {\bf F}_\text{BBT}    \right\|_\text{F}} $\cite{yu2016alternating,zhang2019precoding}. Letting $J({\bf F}_\text{RFT},{\bf F}_\text{BBT}, {\bf F}_\text{BBR}, {\bf F}_\text{RFR})$ denote the objective function of problem (P3), the effectiveness  of the normalization step is shown in the following remark.
\begin{remark}\label{remark1}
If $J({\bf F}_{\rm{RFT}},{\bf F}_{\rm{BBT}}, {\bf F}_{\rm{BBR}}, {\bf F}_{\rm{RFR}}) \le \delta$ when ignoring the power constraint in \eqref{Problem_relay_step2}, $J({\bf F}_{\rm{RFT}}, \hat {\bf F}_{\rm{BBT}}, {\bf F}_{\rm{BBR}}, {\bf F}_{\rm{RFR}}) \le 4\delta$, where $\hat {\bf F}_{\rm{BBT}} =\frac{\sqrt{n_\text{s}}}{    \left\| {\bf F}_\text{RFT} {\bf F}_\text{BBT}    \right\|_\text{F}}   {\bf F}_{\rm{BBT}}$.
\end{remark}
\begin{proof}
The proof of Remark 1 is omitted here since it is similar to that in  \cite{yu2016alternating}. 
\end{proof}
In Remark 1, it demonstrates that after minimizing the objective function of (P3) to a sufficiently  small value $\delta$ when ignoring the power constraint in \eqref{Problem_relay_step2}, the normalization step will also guarantee a small  value $4\delta$ for  minimizing the objective function. 

With above close-form solutions in \eqref{frft}, \eqref{frfr}, \eqref{fbbr} and \eqref{fbbt}, the MM based HBF design for both transmuting and receiving parts at  relay is summarized in Algorithm \ref{ALG2}.

 \begin{algorithm}[t]\label{ALG2}
	\caption{MM-based HBF  algorithm} 
	\begin{algorithmic}[1]
	    \STATE \textbf{Input}: $\hat{\bf F}_\text{opt}$;
      	\STATE \textbf{Output}: ${\bf F}_\text{RFT}$, ${\bf F}_\text{BBT}$, ${\bf F}_\text{RFR}$, ${\bf F}_\text{BBR}$; 
		\STATE \textbf{Initialize}: ${\bf F}_\text{RFT}^{(0)} $,  ${\bf F}_\text{RFR}^{(0)}$, ${\bf F}_\text{BBR}^{(0)}$ and outer iteration $k_\text{o}=0$;   
		\REPEAT
		\STATE  Fix ${\bf F}_\text{RFT}^{(k_\text{o})}$, ${\bf F}_\text{RFR}^{(k_\text{o})}$ and  ${\bf F}_\text{BBR}^{(k_\text{o})}$,  compute    ${\bf F}_\text{BBT}^{(k_\text{o}+1)} $ \\according   to \eqref{fbbt};
		\STATE  Use MM method to compute ${\bf F}_\text{RFT}^{(k_\text{o}+1)}$:
			    	\STATE \textbf{Initialize}: ${\bf F}_\text{RFT}^{(0)}={\bf F}_\text{RFT}^{(k_\text{o})}$,  and inner iteration $k_\text{i}=0$;   
					\REPEAT 
					\STATE   Compute ${\bf F}_\text{RFT}^{(k_\text{i}+1)}$ according to \eqref{frft};
					\STATE   $k_\text{i} \leftarrow k_\text{i}+1$;
					\UNTIL{the stopping criteria is  met.}
					\STATE \textbf{Update}: ${\bf F}_\text{RFT}^{(k_\text{o}+1)}={\bf F}_\text{RFT}^{(k_\text{i})}$;   
		\STATE  Fix ${\bf F}_\text{RFT}^{(k_\text{o}+1)}$, ${\bf F}_\text{BBT}^{(k_\text{o}+1)}$  and  ${\bf F}_\text{RFR}^{(k_\text{o})}$,  compute    ${\bf F}_\text{BBR}^{(k_\text{o}+1)} $ according   to \eqref{fbbr};
		\STATE  Use MM method  to compute ${\bf F}_\text{RFR}^{(k_\text{o}+1)}$:
			    	\STATE \textbf{Initialize}: ${\bf F}_\text{RFR}^{(0)}={\bf F}_\text{RFR}^{(k_\text{o})}$,  and inner iteration $k_\text{i}=0$;   
					\REPEAT
					\STATE   Compute ${\bf F}_\text{RFR}^{(k_\text{i}+1)}$ according to \eqref{frfr};
					\STATE   $k_\text{i} \leftarrow k_\text{i}+1$;
					\UNTIL{the stopping criteria is  met.}	
					\STATE \textbf{Update}: ${\bf F}_\text{RFR}^{(k_\text{o}+1)}={\bf F}_\text{RFR}^{(k_\text{i})}$;   											
		\STATE $k_\text{o} \leftarrow k_\text{o}+1$;
		\UNTIL{the stopping criteria is  met.}
		\STATE Compute ${\bf F}_\text{BBT}= \frac{\sqrt{n_\text{s}}}{ \left\| {\bf F}_\text{RFT} {\bf F}_\text{BBT} \right\|_\text{F} }  {\bf F}_\text{BBT}$. 
	\end{algorithmic} 
\end{algorithm}

\subsection{Convergence analysis and computational complexity}

The general formulation for problem (P2) can be given as
\begin{equation}\label{22}
    \min_{\bm{x},\bm{y}}~F(\bm{x},\bm{y}),~~\text{s.t.}~G(\bm{x},\bm{y})=0,
\end{equation}
where $F(\cdot,\cdot)$ is bi-convex and $G(\cdot,\cdot)$ is bi-affine\footnote{In other words, for any fixed $\bm{x},\bm{y}$, $F(\cdot,\bm{y})$ and $F( \bm{x},\cdot)$ are convex; while $G(\cdot,\bm{y})$ and $G(\bm{x},\cdot)$ are affine.}. According to \cite{boyd}, ADMM can be applied to solve problem (\ref{22}). 
 The convergence of Algorithm 1 is under research and would not be analyzed in this paper. Instead, we will show the convergence of Algorithm 1 by simulation result in Section \ref{Num_alg_per}.
We then summarize the   main complexity   of Algorithm 1 in the following theorem. 

\begin{theorem}
If $N_\text{t}=N_\text{r}= n_\text{t}= n_\text{r}$, the main complexity of Algorithm 1 is  $\mathcal{O}( 2 K_\text{A}  (n_\text{s}^3+1) N_\text{t}^3 )$, where $K_\text{A}$ is the number of iterations.
\end{theorem}
\begin{proof}
For Algorithm 1,  the main complexity in each iteration includes two parts: 

1)  Derive ${\bf F}_\text{opt}$ based on the  singular value decomposition   of two channel matrices (i.e., ${\bf H}_\text{SR}$ and ${\bf H}_\text{RD}$). According to \cite{comon1990tracking},  the main complexity in this part is $\mathcal{O}( \max( N_\text{t}, n_\text{r}) \min( N_\text{t}, n_\text{r})^2 +\max( N_\text{r}, n_\text{t}) \min( N_\text{r}, n_\text{t})^2) $. 

 2) Compute $ {\bf F}_\text{T}$ and $ {\bf F}_\text{R}$ according to \eqref{Problem_FT_solve} and \eqref{Problem_FR_solve}, respectively. The main complexity in this part comes from the inversion operations in \eqref{Problem_FT_solve} and \eqref{Problem_FR_solve}, which is $\mathcal{O}( n_\text{s}^3 (n_\text{t}^3 + n_\text{r}^3))$. 
 
Thus, the main complexity of Algorithm 1  is given by  $\mathcal{O}( K_\text{A}( \max( N_\text{t}, n_\text{r}) \min( N_\text{t}, n_\text{r})^2 +\max( N_\text{r}, n_\text{t}) \min( N_\text{r}, n_\text{t})^2+ n_\text{s}^3 (n_\text{t}^3 + n_\text{r}^3) ))$. 
\end{proof}

The convergence and main complexity of Algorithm 2 are summarized in the following theorem. 
\begin{theorem}\label{theorem_alg2}
The convergence of Algorithm 2 is guaranteed. If $n_\text{t}=n_\text{r}$ and $n_\text{s}=N_\text{RFT}= N_\text{RFR}$, the main complexity of Algorithm 2  is $\mathcal{O}( 2K_\text{out}(K_\text{in} n_\text{t}^3 N_\text{RFT}^3  +  n_\text{t}N_\text{RFT}^2 + n_\text{t}n_\text{s}^2  )   )$, where $ K_\text{out}$ and $K_\text{in}$ are the numbers of outer  and inner iterations, respectively
\end{theorem}
\begin{proof}
See Appendix  \ref{appendix_A}
\end{proof}

 For the hybrid beamformers design at source and destination, we can also utilize MM methods and follow  similar procedures in Algorithm 2, and details are omitted for space limitation.   
Above proposed HBF algorithms are iterative algorithms and suffer from high computational complexity as the number of antennas increases. 
Further,   the  proposed HBF algorithms  and existing optimization based algorithms are linear mapping from the channel matrix and the hybrid beamformers which require  a  real-time computation and are not robust to noisy channel input data. Driven by following  advantages of ML\cite{elbir2019cnn}: (1) low complexity when solving optimization-based problem and (2) capable to extrapolate new features form noisy and limited training data,  we will  propose two learning based approaches to address these problems in the following section.



\section{learning based FD mmwave  beamforming design} 
In this section, we will present our  learning frameworks for HBF design. Firstly, we present the framework of ELM  to design hybrid beamformers and the training data generation  approach for robust HBF design.  Then, we briefly introduce the HBF design  based on CNN.  

\subsection{FD mmwave beamforming design with  ELM}
Feedforward Neural Network (FNN) is a powerful tool for regression and classification \cite{ITP98,17,yy}. As a special FNN, the single layer feedforward network (SLFN) has also been investigated for low complexity \cite{18}. In the SLFN, the weights of input nodes are optimized through the training procedure. Moreover, ELM is developed in \cite{20,21,22}, which consists of only one hidden layer. The weights of input nodes and bias for the hidden nodes are generated randomly. It is shown that ELM can achieve fast running speed with acceptable generalization performance. Since the processing time is one of the bottleneck for low latency communication, using ELM for HBF design can significantly reduce the overall delay. Moreover, due to the hardware constraint (e.g., limited computational capability and memory resources) of mobile terminals, it is easy to implement ELM based component because of its simple architecture.


Thus, in what follows, we will utilize ELM to extract the  features of FD mmWave channels and predict the hybrid beamformers for all nodes. As mentioned in Sec. \ref{FD_mmWave_beamforming_design}, we will design different ELM networks for different nodes. We   mainly focus on the ELM network design for HBF of relay, since the ELM network designs for source and destination  are straightforward following a similar approach to that of relay node.

 We assume that the training dataset is $\mathcal{D}= \left\{(\bm{x}_j,\bm{t}_j)|j=1, \ldots,N \right\} $, where $\bm{x}_j$ and $\bm{t}_j$ are sample and target for the $j$-th training data. Specifically, considering the $j$-th training data, we have the  input as 
  $\bm{x}_j=[\text{Re}(\text{vec}(\overline {\bf{H}}^{(j)}_{\text{RD}})), \text{Im}(\text{vec}(\overline{\bf{H}}^{(j)}_{\text{RD}})),\text{Re}(\text{vec}(\overline{\bf{H}}^{(j)}_{\text{SR}})),$ $\text{Im}(\text{vec}(\overline{\bf{H}}^{(j)}_{\text{SR}})),\text{Re}(\text{vec}(\overline{\bf{H}}^{(j)}_{\text{SI}})),  \text{Im}(\text{vec}(\overline{\bf{H}}^{(j)}_{\text{SI}}))]^T \in\mathbb{R}^{N_\text{I}} $ 
   with dimension ${N_\text{I}}= 2(n_\text{r}(N_\text{t}+N_\text{t} )+N_\text{r}n_\text{t} )$, where $\overline{\bf{H}}^{(j)}_\Omega \in \mathcal{CN}({\bf{H}}_{\Omega},\Gamma_\Omega)$,   $\Omega \in \{\text{SR},\text{RD},\text{SI}   \}$ is the index set for different links.  And $\Gamma_\Omega $ denotes the variance  of additive white Gaussian noise (AWGN),  with its $(m,n)$-th entry as $[\Gamma_\Omega ]_{m,n} = |[\overline{\bf{H}}^{(j)}_\Omega ]_{m,n}|^2 -\text{SNR}_{\text{Train}} $ (dB), where $\text{SNR}_{\text{Train}}$ is the SNR for the training data\cite{elbir2019cnn}.    
    The target of $j$-th data is  $\bm{t}_j=[\text{Re}(\text{vec}({\bf{F}}^{(j)}_{\text{BBT}})),\text{Im}(\text{vec}({\bf{F}}^{(j)}_{\text{BBT}})),\text{Re}(\text{vec}({\bf{F}}^{(j)}_{\text{BBR}})),$ $\text{Im}(\text{vec}({\bf{F}}^{(j)}_{\text{BBR}})),\arg(\text{vec}({\bf{F}}^{(j)}_{\text{RFT}})),\arg(\text{vec}({\bf{F}}^{(j)}_{\text{RFR}}))]  \in\mathbb{R}^{N_\text{o}}$ with dimension $ N_\text{o}= n_\text{t} N_\text{RFT}+n_\text{r} N_\text{RFR}+ 2N_\text{s}(N_\text{RFR}+N_\text{RFT})$, which is from the corresponding beamformers obtained by Algorithms 1 and 2 with input $\overline{\bf{H}}^{(j)}_\Omega$. The ELM with $L$ hidden nodes and activation function $g(x)$ is shown in Fig. \ref{fig_ELM}. The blocks in input layer and output layer consist of neurons which have the number of the dimensions of corresponding input and output, respectively. According to  \cite{22}, the output of ELM related to sample $\bm{x}_j$ can be mathematically modeled as
\begin{equation}\label{eq18}
\sum_{i=1}^{L}\beta_ig_i(\bm{x}_j)=\sum_{i=1}^{L}\beta_ig(\bm{w}_i^T\bm{x}_j+b_i)={\bf g}(\bm{x}_j)\bm{\beta},
\end{equation} 

\begin{figure}[t] 
	\vskip 0.2in
	\begin{center}
		\centerline{\includegraphics[width=85mm]{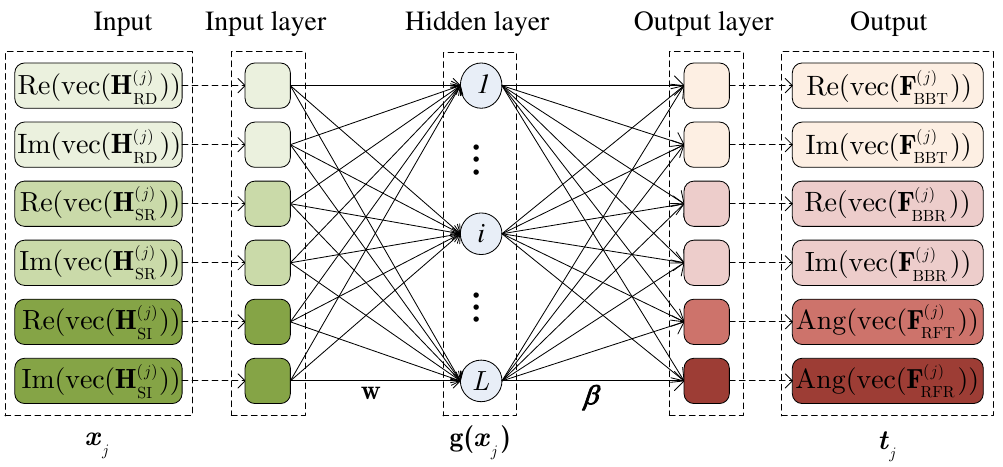}}
		\caption{ELM network for HBF design at relay.}
		\label{fig_ELM}
	\end{center}
	\vskip -0.2in
\end{figure}
where $\bm{w}_i=[w_{i,1},\ldots,w_{i, N_\text{I}}]^T$ is the weight vector connecting the $i$-th hidden node and the input nodes, $\bm{\beta}=\left[\beta_1,\ldots, \beta_L\right]^T \in \mathbb{R}^{L\times N_\text{o}}$, and $\beta_i=[\beta_{i,1},\ldots,\beta_{i,N_\text{o}}]^T$ is the weight vector connecting the $i$-th hidden node and the output nodes, and $b_i$ is the bias of the $i$-th hidden node. Considering all the samples in $\mathcal{D}$, we stack (\ref{eq18}) to obtain the hidden-layer output as

\begin{equation}
{\bf G}=\left[\begin{matrix}
{\bf g}(\bm{x}_1)\\\vdots\\ {\bf g}(\bm{x}_N)
\end{matrix} \right]=\left[\begin{matrix}
g_1(\bm{x}_1) & \cdots &g_L(\bm{x}_1) \\
\vdots & \cdots & \vdots \\
g_1(\bm{x}_N) & \cdots &g_L(\bm{x}_N)
\end{matrix}\right]_{L\times N}.
\end{equation}
Actually, we can regard ${\bf G}$ as the feature mapping from the training data, which maps the data from the $N_\text{I}$-dimensional space into the $L$-dimensional hidden-layer feature space.

 \begin{figure}[t] 
	\vskip 0.2in
	\begin{center}
		\centerline{\includegraphics[width=88mm]{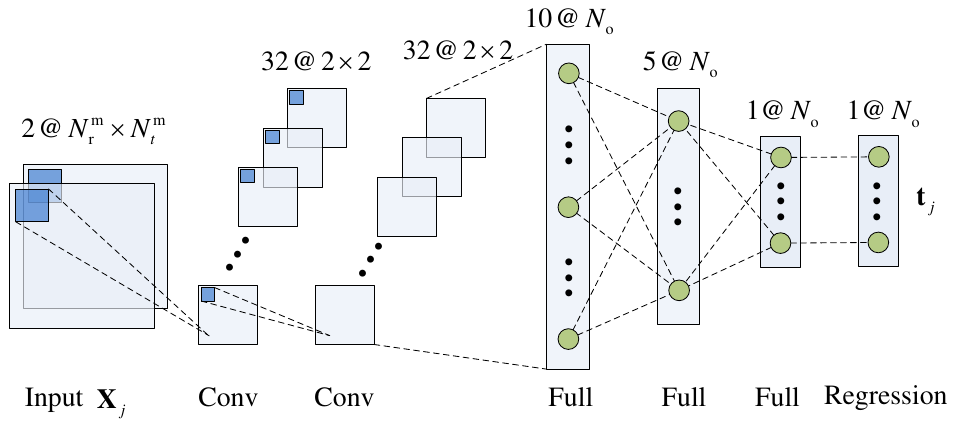}}
		\caption{CNN network for HBF design at relay.}
		\label{fig_CNN}
	\end{center}
	\vskip -0.2in
\end{figure}

\begin{figure*}[ht]
	\centering
	\subfloat[ ]{\includegraphics[width =.32\textwidth]{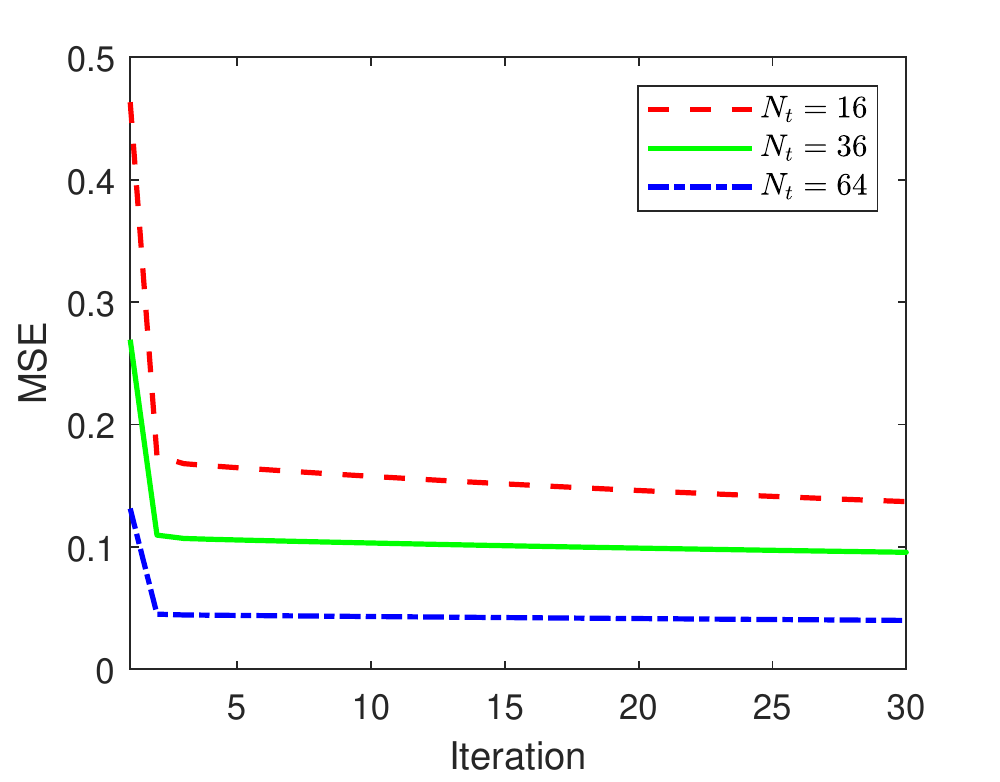}\label{fig:4a}}
	\subfloat[ ]{\includegraphics[width=.32\textwidth]{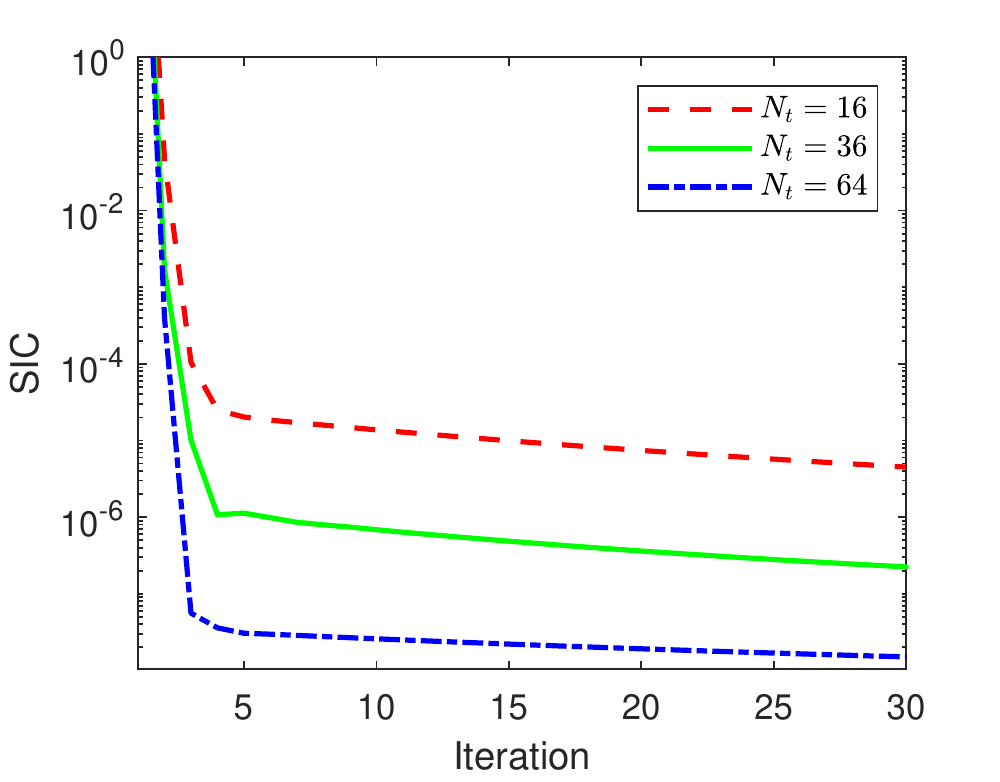}\label{fig:4b}}
		\subfloat[]{\includegraphics[width=.32\textwidth]{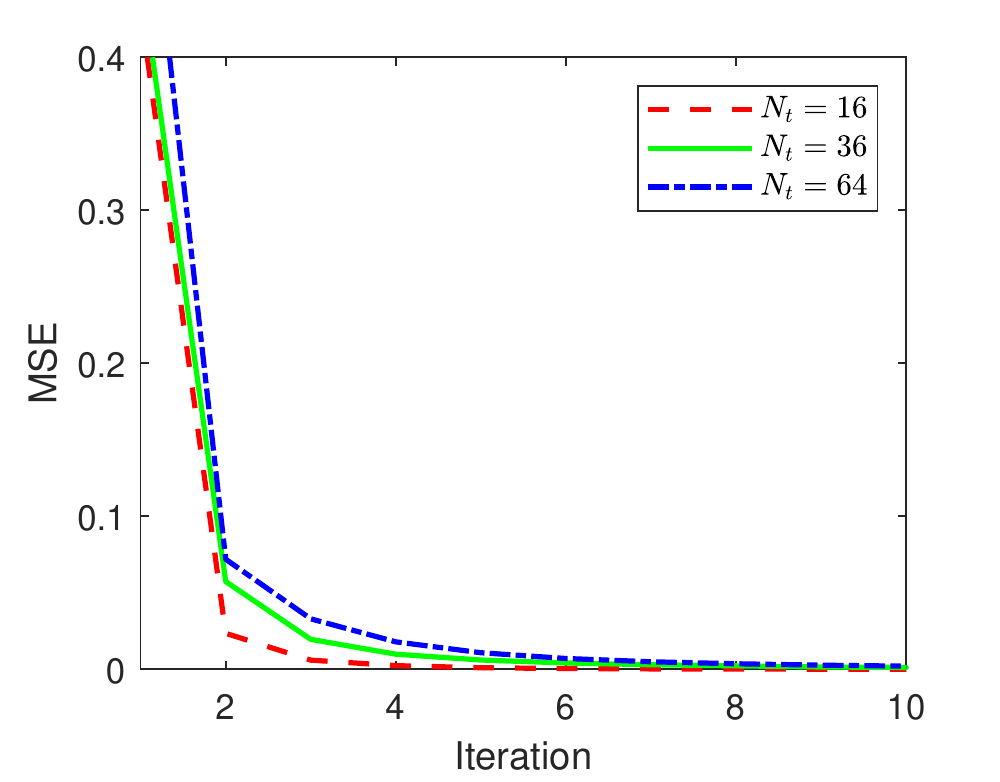}\label{fig:4c}}
	\caption{MSE and SIC performance  for the proposed beamforming algorithms vs the number of iterations and  different numbers of antennas, $N_t=N_r$, $N_s=2$ and  $N_\text{RFR}=N_\text{RFT}=4$: (a) MSE for ADMM-based  algorithms; (b) SIC performance for  ADMM-based  algorithms; (c)  MSE for MM-based  algorithms.}
	\label{fig:4}
\end{figure*}

Since there is only one hidden layer in ELM, with randomized weights $\{\bm{w}_i\}$ and biases $\{b_i\}$, the goal is to tune the output weight $\bm{\beta}$ with training data $\mathcal{D}$ through minimizing the ridge regression problem 
\begin{equation} 
(\text{P}7):~ \bm{\beta}^*=\arg\min_{\bm{\beta}}~\frac{\lambda}{2}\left\| {\bf G}\bm{\beta}-{\bf T}\right\|^2 + \frac{1}{2} \left\| \bm{\beta} \right\|^2,
\end{equation}
where $ {\bf T}=\left[\bm{t}_{1},\ldots, \bm{t}_N\right]^T_{N\times N_\text{o}}$ is the concatenated target, $\lambda$ is the trade-off parameter between the training error and the regularization. According to  \cite{21}, the closed-form solution for (P7) is 
\begin{equation}\label{beta1}
\bm{\beta}^*={\bf G}^T\left(\frac{\bf{I}}{\lambda}+{\bf G}{\bf G}^T \right)^{-1}{\bf T},~~N\leq L,
\end{equation}
or
\begin{equation}\label{beta2}
\bm{\beta}^*=\left(\frac{\bf{I}}{\lambda}+{\bf G}^T {\bf G}\right)^{-1}{\bf G}^T{\bf T},~~N>L, 
\end{equation}
where $\bm{\beta}^*$ in \eqref{beta1} is derived for the case where the number of training samples is small, while  $\bm{\beta}^*$ in \eqref{beta2} is derived for the case where the number of training samples is huge. 
From above, we can see that ELM is with very low complexity since there is only one layer's parameters to be trained and the weight of output layer (i.e., $\bm{\beta}$ ) is given in closed-form. 


\subsection{FD mmwave beamforming design with  CNN} 
Due to the  advantages of data compression, CNN is another promising learning network to solve communication problems at the physical layer. Some latest CNN based  HBF designs are presented in \cite{elbir2019cnn,bao2020deep}, but these works are only for single-hop wireless communications. Based on the CNN-based hybrid beamforming  model in \cite{elbir2019cnn,bao2020deep}, we extend it to FD mmWave systems. The CNN-based architecture is shown in Fig.~\ref{fig_CNN}, which has a total eleven layers, including an input layer, two convolution layers, three fully connected layers, a regression output layer and  four activation layers after each  convolutional layer and fully connected layer. Detailed parameters in each layer are shown in Fig.~\ref{fig_CNN}. 
Different from ELM, the $j$-th input data ${\bf{X}}_j$ of CNN is a  three-dimensional (3D) real matrix  with  size $N_\text{r}^\text{m} \times N_\text{t}^\text{m} \times2 $ 
where $N_\text{t}^\text{m} = \max(N_t, n_t) $ and $N_\text{r}^\text{m} = \max(N_r, n_r) $. We define the first channel of the input as the element-wise real value of the input channel matrix given by $ [ {\bf{X}}_j]_{:,:,1} =[\text{Re}(\overline {\bf{H}}^{(j)}_{\text{SR}}),
 \text{Re}(\overline {\bf{H}}^{(j)}_{\text{RD}}), \text{Re}(\overline {\bf{H}}^{(j)}_{\text{SI}}) ]$, and  the second channel of the input as the element-wise imaginary value of the input channel matrix given by $ [ {\bf{X}}_j]_{:,:,2} =[\text{Im}(\overline {\bf{H}}^{(j)}_{\text{SR}}), \text{Im}(\overline {\bf{H}}^{(j)}_{\text{RD}}), \text{Im}(\overline {\bf{H}}^{(j)}_{\text{SI}}) ]$. The output of the  CNN is the same as that of ELM, which can be obtained from   Algorithm \ref{ALG2}. More details of CNN  can be found in \cite{elbir2019cnn}.

\section{Numerical simulations} 
In this section, we will numerically evaluate  the performance of the proposed MO  and ADMM based HBF algorithm (MM-ADMM-HBF),  ELM-based HBF method (ELM-HBF) and CNN-based HBF method (CNN-HBF). We compared our results with four benchmark algorithms:  SI-free fully  digital  beamforming (Full-D), fully digital beamforming with SI (Full-D with SI), HD fully digital beamforming (HD Full-D) and OMP-based HBF method (OMP-HBF) \cite{zhang2019precoding}. 
The channel parameters are set to $N_\text{c}= 5$, $N_\text{p}= 10$, $d=\frac{\lambda}{2}$ and $\alpha_{k,l} \sim \mathcal{CN}(0,1)$\cite{zhang2019precoding}.    The bandwidth  of this system is $2$~GHz with central carrier frequency $f_c=28$~GHz. According to \cite{zhang2019precoding},  the pathloss is  $P_\text{loss}= 61.5+20\log(r)+\varepsilon $ (dB) where $\varepsilon \sim N(0,5.8) $ and $r$ denotes the distance between  transmitter and receiver. We assume that the distance between source and relay, and  the distance between relay and destination are $r_\text{sr}=100$~m and  $r_\text{rd}=100$~m, respectively. We assume that all nodes in FD mmWave systems have the same hardware constraints, and $N_\text{t}=n_\text{t}$, $N_\text{r}=n_\text{r}$, $N_\text{RFR}=N_\text{RFD}$, $N_\text{s}=n_\text{s}$ and  $N_\text{RFT}=N_\text{RFS}$.  In both training  and testing stages,  each channel realization is added by AWGN with different powers of $\text{SNR}_{\text{Train}}=\text{SNR}_{\text{Test}} \in \{15, 20, 25\}$ dB.

\subsection{Performance of ADMM and MM based beamforming}\label{Num_alg_per}

 Fig.~\ref{fig:4} summarizes the performance of proposed beamforming algorithms versus the number of iterations and  different numbers of antennas.
 Fig.~\subref*{fig:4a} shows the MSE performance (i.e., $\left\| {\bf F}_\text{opt} -    {\bf F}_\text{T} {\bf F}_\text{R}^{H} \right\|_\text{F}^2$) of the ADMM based beamforming  algorithm (Algorithm 1).  We can observe a fast convergence of the proposed algorithm and the convergence rate decreases as the number of antennas increases. Results also  show that  the  MSE of Algorithm 1 at convergence   decreases  as the number of antennas increases. 
 Fig.~\subref*{fig:4b} shows the   SIC performance (i.e., $ {\|{\bf F}_{\rm R}^{H} {\bf H}_{\rm SI} {\bf F}_{\rm T} \|}_{\rm{F}}^2$)  of Algorithm 1.
It is shown  that the power of the SI decreases with increasing numbers of algorithmic iteration. We can also see that using a large number of antennas can eliminate SI faster.
  Fig.~\subref*{fig:4c}  shows the  MSE   performance (i.e., $\left\| \hat{\bf F}_\text{opt} -  {\bf F}_\text{RFT} {\bf F}_\text{BBT} {\bf F}_\text{BBR}^{H} {\bf F}_\text{RFR}^{H} \right\|_\text{F}^2$) of the MM-based HBF algorithm (Algorithm 2). It is shown that a very fast convergence rate of Algorithm 2, even for the case with  a large number of antennas (e.g., $N_t=64$). 

\begin{figure}[t] 
	\vskip 0.2in
	\begin{center}
		\centerline{\includegraphics[width=82mm]{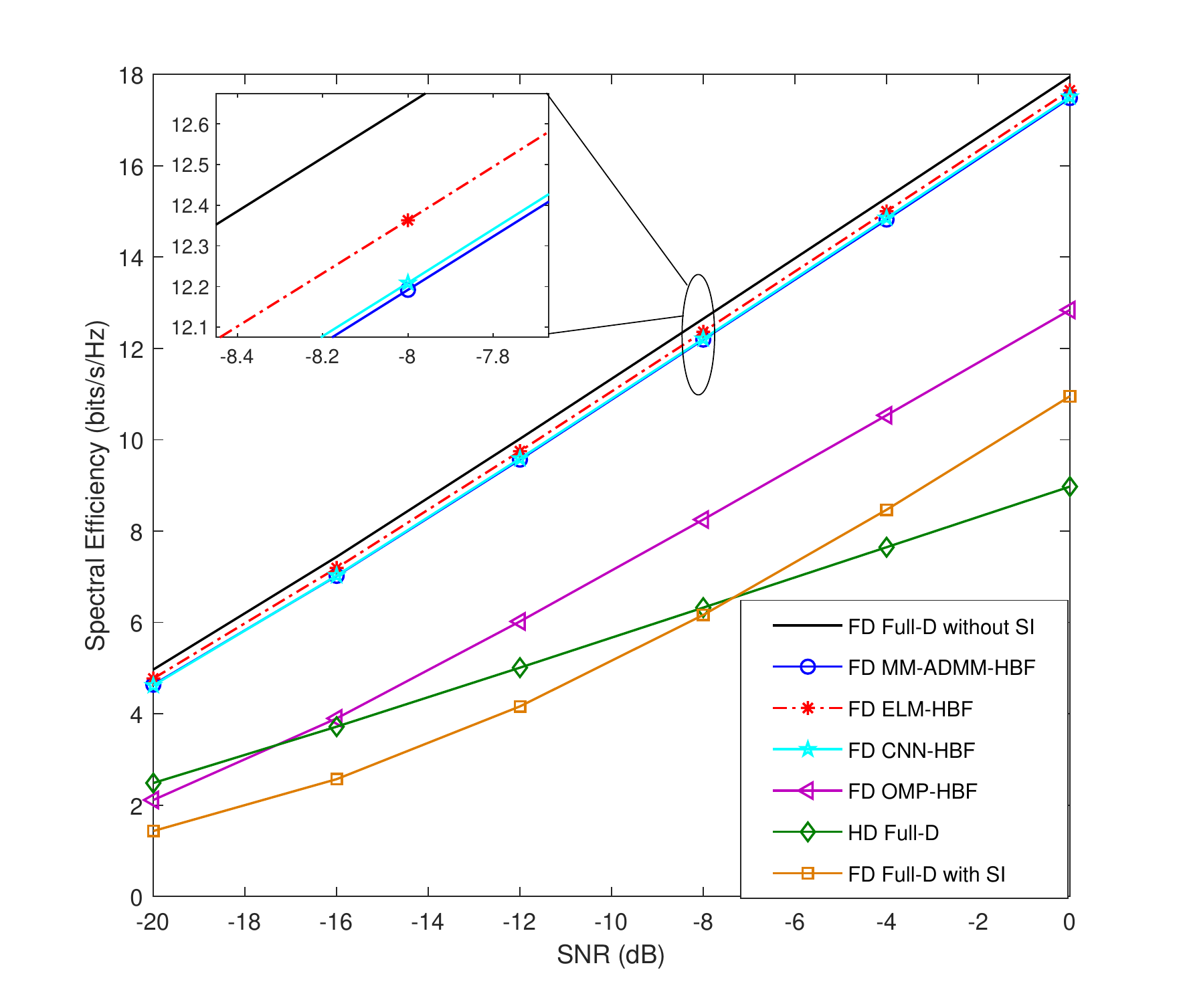}}
		\caption{Spectral efficiency of various HBF algorithms vs SNR with $N_t=N_r=36$,  $N_\text{RFR}=6$, $N_\text{RFT}=8$ and $N_\text{s}=2$.}
		\label{fig_SE_all_1a}
	\end{center}
	\vskip -0.2in
\end{figure}
 \begin{figure}[t] 
	\vskip 0.2in
	\begin{center}
		\centerline{\includegraphics[width=82mm]{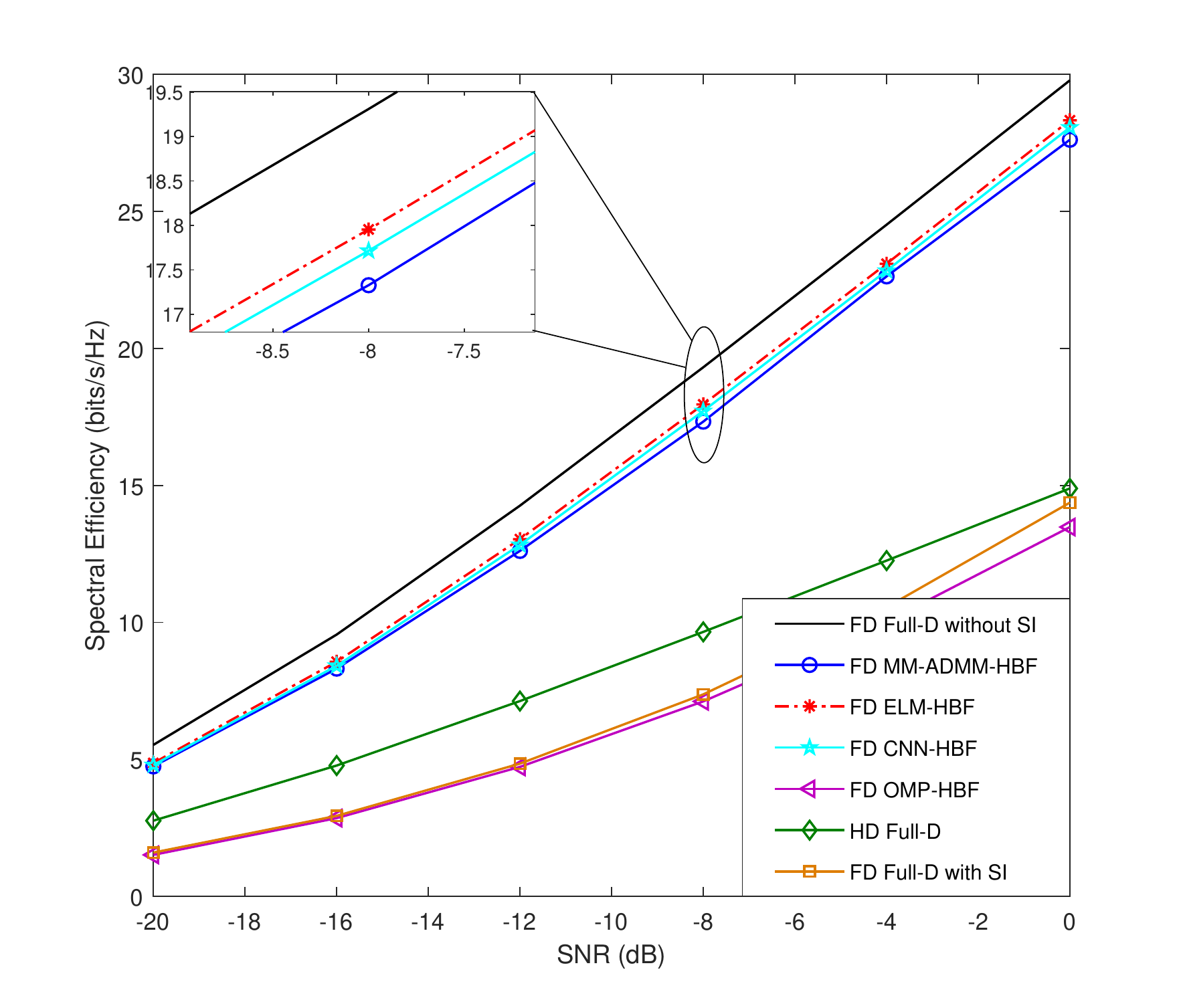}}
		\caption{Spectral efficiency of various HBF algorithms vs SNR with $N_t=N_r=36$,  $N_\text{RFR}=6$, $N_\text{RFT}=8$ and $N_\text{s}=4$.}
		\label{fig_SE_all_1b}
	\end{center}
	\vskip -0.2in
\end{figure}

\subsection{Hybrid beamforming performance}

Fig.~\ref{fig_SE_all_1a} and Fig.~\ref{fig_SE_all_1b}  show the spectral efficiency of the proposed HBF methods and  existing methods versus SNR and different numbers of transmitting  streams.  From Fig.~\ref{fig_SE_all_1a},  we can see that the proposed MM-ADMM based HBF algorithm can approximately achieve the performance of fully digital beamforming without SI, which means that  the proposed algorithm can achieve near-optimal HBF and guarantee efficient SIC. We can also see that the proposed algorithm significantly outperforms OMP based algorithm. The reason is that the OMP-based algorithm in \cite{zhang2019precoding} eliminates SI by adjusting the derived  optimal beamformers, which will significantly degrade the spectral efficiency. Furthermore, the  proposed CNN-based and ELM-based HBF methods outperform  other methods. The performance of learning based methods is attributed to extracting the features of noisy input data (i.e., imperfect channels for different links) and be robust to the imperfect channels. We can see that all proposed methods can approximately achieve twice the spectral efficiency to the HD system. 
Fig.~\ref{fig_SE_all_1b} shows that our proposed methods  can also achieve high spectral efficiency with increasing number of transmitting  streams.  However, the spectral efficiency of OMP-based algorithm becomes even lower than  the FD fully-digital beamforming with SI. The reason is that OMP-based algorithm  can  eliminate SI only when $N_\text{RFT}\ge N_\text{RFR}+N_\text{s}$. Comparing the result in  Fig.~\ref{fig_SE_all_1a} to that in  Fig.~\ref{fig_SE_all_1b}, we can find that the spectral efficiency is significantly increased as the number of transmitting  streams increases.

Fig.~\ref{fig_Nt} shows two groups of  spectral efficiency with different numbers of antennas. In each group, simulation results of three proposed methods together with  OMP-HBF and Full-D beamforming methods are presented. It is shown  that the spectral efficiency increases as the number of antennas increases, and the gap of the results within a group decreases simultaneously. Moreover,  the proposed methods can achieve higher spectral efficiency than the OMP-HBF method. We can also see that the  proposed methods can  approximately achieve the performance of Full-D beamforming without SI when $N_\textbf{t}=N_\textbf{r}=64$. Finally, the proposed learning based HBF methods (i.e., ELM-HBF and CNN-HBF)  outperform the optimization-based HBF methods (i.e., OMP-HBF and MM-ADMM-HBF). 
 \begin{figure}[t] 
	\vskip 0.2in
	\begin{center}
		\centerline{\includegraphics[width=82mm]{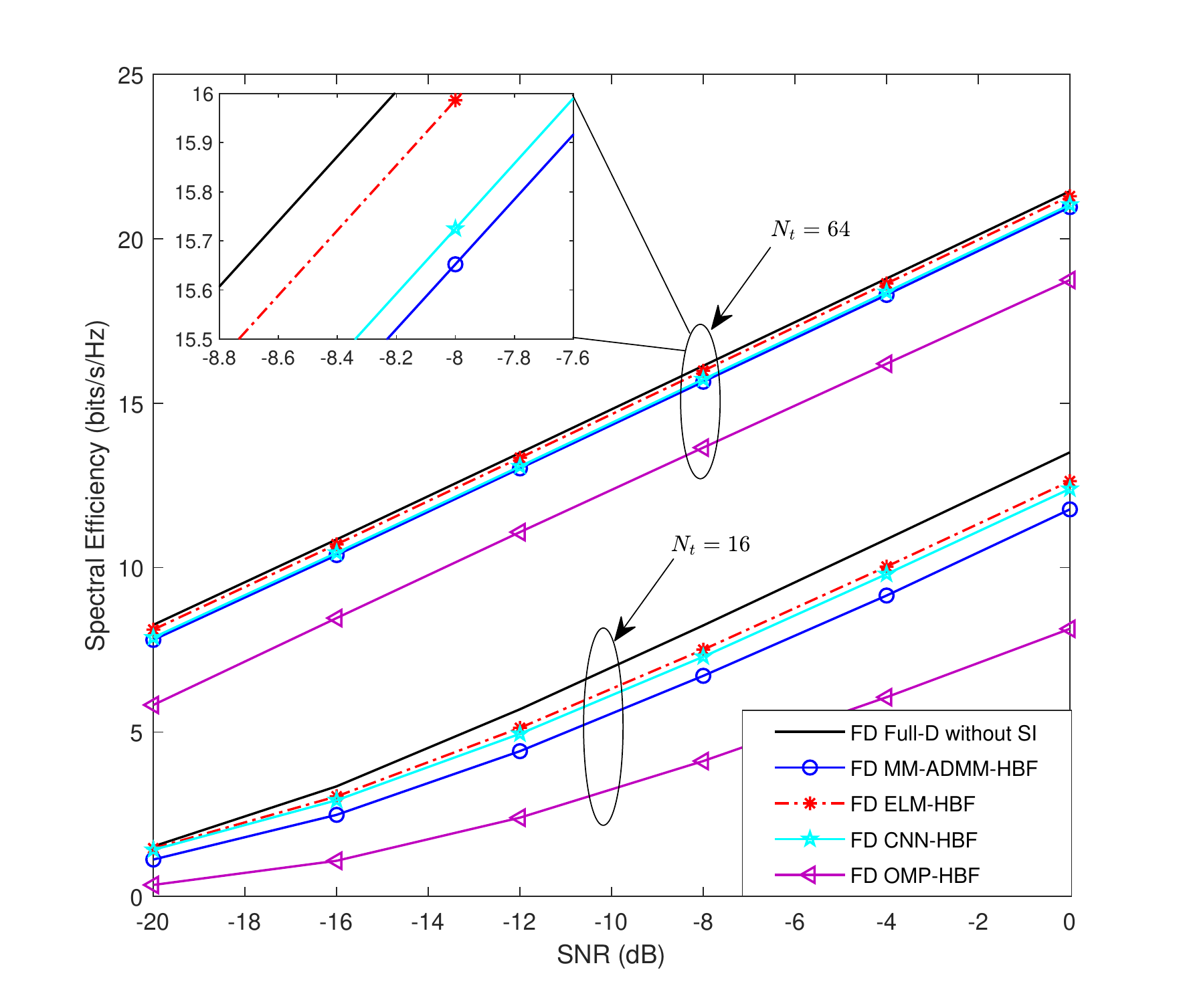}}
		\caption{Spectral efficiency of various HBF algorithms vs SNR and different numbers of antennas with $N_t=N_r$,  $N_\text{RFR}=4$, $N_\text{RFT}=6$.}
		\label{fig_Nt}
	\end{center}
	\vskip -0.2in
\end{figure}

 \begin{figure}[t] 
	\vskip 0.2in
	\begin{center}
		\centerline{\includegraphics[width=82mm]{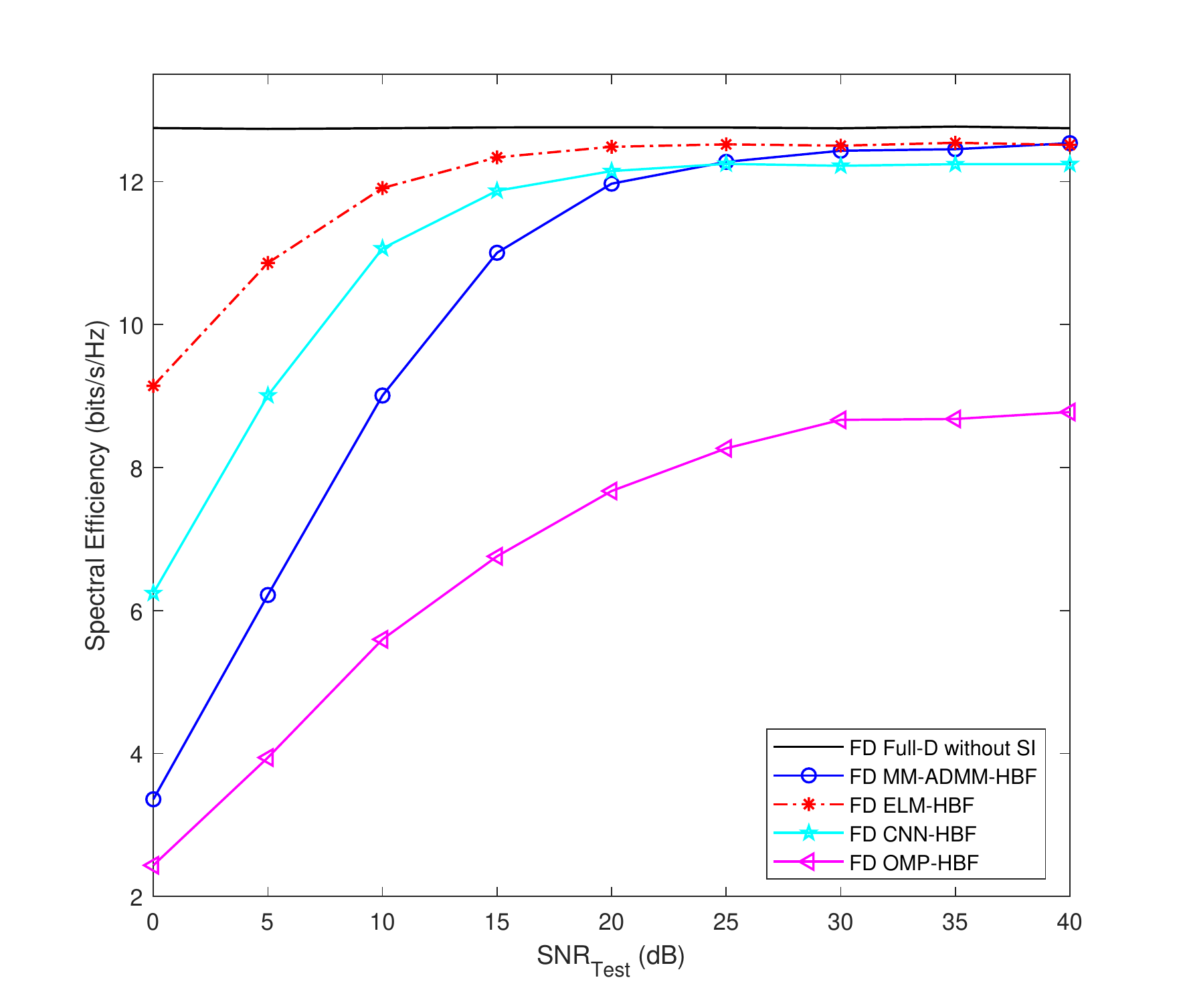}}
		\caption{Spectral efficiency of various HBF algorithms vs  $\text{SNR}_{\text{Test}}$  with $N_t=N_r=36$,  $N_\text{RFR}=4$, $N_\text{RFT}=6$ and $\text{SNR}=-8$ dB.}
		\label{fig_robust}
	\end{center}
	\vskip -0.2in
\end{figure}

\begin{table*}[t]
\caption{Spectral Efficiency (bits/s/Hz), Training and Prediction Time Comparison}\label{complexity_SE}
\scalebox{0.64}{
\begin{tabular}{|c|c|c|c|c|c|c|c|c|c|c|c|c|c|c|c|c|}
\hline
\multirow{3}{*}{Nt} & \multicolumn{2}{c|}{\multirow{2}{*}{MM-ADMM}}                                & \multicolumn{2}{c|}{\multirow{2}{*}{OMP}}                               & \multicolumn{3}{c|}{\multirow{2}{*}{CNN}}                                                                                             & \multicolumn{9}{c|}{ELM}                                                                                                                    \\ \cline{9-17} 
                    & \multicolumn{2}{c|}{}                                                   & \multicolumn{2}{c|}{}                                                   & \multicolumn{3}{c|}{}                                                                                                                 & \multicolumn{3}{c|}{Sigmoid Node}                                                                                                     & \multicolumn{3}{c|}{Multiquadrics RBF Node}                                                                                                    & \multicolumn{3}{c|}{PReLU Node}                                                                                                        \\ \cline{2-17} 
                    & \begin{tabular}[c]{@{}c@{}}Prediction\\ Time (s)\end{tabular} & SE      & \begin{tabular}[c]{@{}c@{}}Prediction\\ Time (s)\end{tabular} & SE      & \begin{tabular}[c]{@{}c@{}}Prediction\\ Time (s)\end{tabular} & \begin{tabular}[c]{@{}c@{}}Training\\ Time (s)\end{tabular} & SE      & \begin{tabular}[c]{@{}c@{}}Prediction\\ Time (s)\end{tabular} & \begin{tabular}[c]{@{}c@{}}Training\\ Time (s)\end{tabular} & SE      & \begin{tabular}[c]{@{}c@{}}Prediction\\ Time (s)\end{tabular} & \begin{tabular}[c]{@{}c@{}}Training\\ Time (s)\end{tabular} & SE               & \begin{tabular}[c]{@{}c@{}}Prediction\\ Time (s)\end{tabular} & \begin{tabular}[c]{@{}c@{}}Training\\ Time (s)\end{tabular} & SE      \\ \hline
16                  & 0.1920                                                        & 6.6393  & 0.0291                                                        & 4.1733  & 0.0063                                                        & 315.51                                                      & 6.8146  & 0.0626                                                        & 66.98                                                       & 7.1410  & 0.0249                                                        & 59.59                                                       & \textbf{7.1491}  & \textbf{0.0056}                                               & \textbf{1.5791}                                             & 7.1291  \\ \hline
36                  &  0.6112                                                       & 11.8715 & 0.3103                                                        & 8.1470  & 0.0134                                                        & 963.58                                                      & 11.8818 & 0.1859                                                        & 275.23                                                      & 12.2395 & 0.1338                                                        & 213.99                                                      & \textbf{12.3890} & \textbf{0.0126}                                               & \textbf{2.9624}                                             & 12.3569 \\ \hline
64                  &  1.7422                                                       & 15.5430 & 1.2210                                                        & 11.5962 & \textbf{0.0349 }                                                       & 3714.3                                                      & 15.7585 & 0.6681                                                        & 1225.6                                                      & 15.8468 & 0.5830                                                        & 1146.1                                                      & \textbf{15.9355} &  0.0530                                                & \textbf{6.1633}                                             & 15.8555 \\ \hline
\end{tabular} }
\end{table*}

In order to evaluate the performance of algorithms on the robustness, we  present the  spectral efficiency of various HBF algorithms versus different noise levels (i.e., $\text{SNR}_{\text{Test}}$) in Fig.~\ref{fig_robust}. Note that SI-free fully  digital  beamforming (Full-D)  is fed with perfect CSI which can achieve the best performance. From  Fig.~\ref{fig_robust}, we can see that the performance of all methods increase with increasing $\text{SNR}_{\text{Test}}$. 
We can also see that both ELM-HBF and CNN-HBF are more robust against the corruption in the  channel data compared to other methods. The reason is that proposed learning based methods estimate the beamformers by extracting the features of noisy  input  data, while  MM-ADMM and OMP methods require optimal digital  beamformers which are derived from noisy channels. Furthermore, it is shown that ELM-HBF outperforms CNN-HBF. The reason is that  the optimal weight matrix of  ELM network can be  derived in a close-form,  while the multi-layer parameters of CNN are  hard to be optimized. Finally, we can see that the ELM-HBF can approximately achieve optimal as $\text{SNR}_{\text{Test}}$ increases.

\subsection{Computational Complexity}
In this part, we measure the computation time of our proposed HBF approaches and compared them with OMP-HBF.  The computation time of a learning machine includes  offline training time and online prediction time. 
Since the learning network performance and  training time are relative to the activation function in the hidden node, we make a performance comparison among following three common  activation functions for ELM:  

(1) Sigmoid  function 
\begin{equation}
g(\bm{w},\bm{x},b)=\frac{1}{1+\exp (-\bm{w}^T\bm{x}-b )};
\end{equation}

(2) Multi-quadratic radial basis function (RBF) 
\begin{equation}
g(\bm{w},\bm{x},b)=\sqrt{\left\| \bm{x}- \bm{w}\right\|^2 +b^2};
\end{equation}

(3) Parametric rectified linear unit (PReLU) function
\begin{equation}
g(\bm{w},\bm{x},a)=\max (0,\bm{w}^T\bm{x} )+ a \min (0,\bm{w}^T\bm{x} ).
\end{equation}
For CNN,  the multi-layer structure  will lead to high computational complexity and  
a simple rectified linear unit (ReLU) activation function (i,e, $g(\bm{w},\bm{x})=\max(0,\bm{w}^T\bm{x})$) is commonly used to reduce the training complexity. Results in  \cite{bao2020deep,elbir2019joint,elbir2019cnn} show that CNN with ReLU can  achieve good classification performance. Thus,  we  consider ReLU activation function for CNN.  We select $N_\text{s}=2$, $N_\text{RFR}=4$,  $N_\text{RFT}=6$ and  $\text{SNR}=-8$ dB. $1000$ channel samples for $10$ channel realizations are fed into the learning machines, and $100$ channel samples are used for testing. 
 For different approaches, we summarize the spectral efficiency (SE), training time and prediction time in Table \ref{complexity_SE}.

We can see that the proposed ELM-HBF and CNN-HBF methods can achieve higher spectral efficiency and less prediction time than the optimization-based methods (i.e., MM-ADMM and OMP). In addition, we can observe that the prediction time increases with   the number of antennas. It is shown that CNN and    ELM with PReLU can achieve  very low  prediction time (e.g.,  less than $0.06$ s for the case with $N_\text{t}=64$).  
Although ELM with multi-quadric RBF can achieve a slightly higher spectral efficiency than that with PReLU, it requires almost ten times the prediction time and  a hundred times the training time  compared to that with PReLU. For instance, the training time of multi-quadric RBF is about $1146.1$ s while it is about $6.1633$ s of PReLU for the case with $N_\textbf{t}=64$. Results show that   CNN always spends longer training time  and achieves lower spectral efficiency than ELM. For instance, CNN takes about 600 times the training time compared to ELM with PReLU for the case with $N_\textbf{t}=64$.

\section{Conclusions}
We proposed two learning schemes for HBF design of FD mmWave systems, i.e., ELM-HBF and CNN-HBF. The learning machines use  noisy channels of different nodes as inputs and output the hybrid beamformers. To provide accurate  labels of input channel data, we first proposed an ADMM based algorithm to achieve SIC beamforming, and then proposed an MM based algorithm for joint transmitting and receiving HBF optimization. The  convergence and  complexity for both algorithms were analyzed. The effectiveness of the proposed methods was evaluated through several experiments.  Results illustrate that both ADMM  and MM based   algorithms can converge  and the SI can be effectively suppressed. Results also show that both proposed ELM-HBF and CNN-HBF methods can   achieve higher spectral efficiency and much lower prediction time than the convectional optimization-based methods. Furthermore,  the proposed learning based methods can achieve more robust HBF performance than conventional methods. In addition,  ELM-HBF with PReLU activation function can achieve much lower training time than that with Sigmoid or RBF activation function.  Since ELM-HBF can achieve much lower computation time and more robust HBF performance than CNN-HBF, it might be more efficient to use  ELM-HBF  for practical implementation. 

\begin{appendices}
   \section{ Proof of Theorem \ref{theorem_alg2}}\label{appendix_A}
 
   To prove the  convergence  of Algorithm 2, we first analyze the convergence of the MM algorithm when calculating ${\bf F}_\text{RFT}^{(k_\text{o}+1)}$ in step 6.  According to the majorizer of  $J({\bf F}_\text{RFT}; {\bf Y}_\text{T} )$ in \eqref{J_T_major}, we have the following four properties, 
   \begin{subequations}
   	\begin{align}	
   	 J\left({\bf F}_\text{RFT}^{\left(k_\text{i}\right)}; {\bf Y}_\text{T}^{(k_\text{o})}  \right)& = \bar J\left({\bf F}_\text{RFT}^{(k_\text{i})}; {\bf Y}_\text{T}^{(k_\text{o})},{\bf F}_\text{RFT}^{(k_\text{i})} \right), \label{Property1}\\
   	 \nabla_{{\bf F}_\text{RFT}} J\left({\bf F}_\text{RFT}; {\bf Y}_\text{T}^{(k_\text{o})}  \right) &=\nabla_{{\bf F}_\text{RFT}} \bar J\left({\bf F}_\text{RFT}; {\bf Y}_\text{T}^{(k_\text{o})},{\bf F}_\text{RFT}^{(k_\text{i})} \right), \label{Property2}\\
   	 J\left({\bf F}_\text{RFT}^{(k_\text{i}+1)}; {\bf Y}_\text{T}^{(k_\text{o})}  \right)& \mathop  \le \limits^{(a)}  \bar J\left({\bf F}_\text{RFT}^{(k_\text{i}+1)}; {\bf Y}_\text{T}^{(k_\text{o})},{\bf F}_\text{RFT}^{(k_\text{i})} \right), \label{Property3}\\
   	 \bar J\left({\bf F}_\text{RFT}^{(k_\text{i}+1)}; {\bf Y}_\text{T}^{(k_\text{o})},{\bf F}_\text{RFT}^{(k_\text{i})} \right) & \mathop  \le \limits^{(b)}  \bar J\left({\bf F}_\text{RFT}^{(k_\text{i})}; {\bf Y}_\text{T}^{(k_\text{o})},{\bf F}_\text{RFT}^{(k_\text{i})} \right), \label{Property4}
   	\end{align}	
   \end{subequations}
   where  (a) follows from   $\lambda_\text{T}{\bf I} \ge {\bf Q}_\text{T}$, (b) follows from  $ \bar J({\bf F}_\text{RFT}^{(k_\text{i}+1)}; {\bf Y}_\text{T}^{(k_\text{o})},{\bf F}_\text{RFT}^{(k_\text{i})} )= \mathop{\min}\limits_{{\bf F}_\text{RFT}}  \bar J({\bf F}_\text{RFT}; {\bf Y}_\text{T}^{(k_\text{o})},{\bf F}_\text{RFT}^{(k_\text{i})} )$, and $ {\bf Y}_\text{T}^{(k_\text{o})}={\bf F}_\text{BBT}^{(k_\text{o}+1)} {\bf F}_\text{BBR}^{H(k_\text{o})} {\bf F}_\text{RFR}^{H(k_\text{o})}$. Based on properties \eqref{Property1}, \eqref{Property3} and \eqref{Property4}, we obtain  
   \begin{equation}\label{inequality1}
   	\begin{split}	
   J\left({\bf F}_\text{RFT}^{(k_\text{i}+1)}; {\bf Y}_\text{T}^{(k_\text{o})} \right) &\le \bar J\left({\bf F}_\text{RFT}^{(k_\text{i}+1)}; {\bf Y}_\text{T}^{(k_\text{o})},{\bf F}_\text{RFT}^{(k_\text{i})} \right) \\&\le \bar J\left({\bf F}_\text{RFT}^{(k_\text{i})}; {\bf Y}_\text{T}^{(k_\text{o})},{\bf F}_\text{RFT}^{(k_\text{i})} \right)\\&=J\left({\bf F}_\text{RFT}^{(k_\text{i})}; {\bf Y}_\text{T}^{(k_\text{o})} \right).  
   	\end{split}	
   \end{equation}
   Thus, $\{ J({\bf F}_\text{RFT}^{(k_\text{i})}; {\bf Y}_\text{T}^{(k_\text{o})} )\}$ is a non-increasing  sequence and thus it converges since $J({\bf F}_\text{RFT}; {\bf Y}_\text{T} )$ is lower bounded. Further, since  $J({\bf F}_\text{RFT}; {\bf Y}_\text{T}^{(k_\text{o})} )$ and $ \bar J({\bf F}_\text{RFT}; {\bf Y}_\text{T}^{(k_\text{o})},{\bf F}_\text{RFT}^{(k_\text{i})} )$ have the same gradient at point ${\bf F}_\text{RFT}^{(k_\text{i})} \in \mathcal{F}_\text{RFT}$ according to \eqref{Property2}, ${\bf F}_\text{RFT}^{(k_\text{i})}$ can converge to a stationary point solution of original problem (P4). After the converge of step 6, we have $J({\bf F}_\text{RFT}^{(k_\text{o}+1)}; {\bf Y}_\text{T}^{(k_\text{o})} ) \le J({\bf F}_\text{RFT}^{(k_\text{o})}; {\bf Y}_\text{T}^{(k_\text{o})} )$.
   Similarly, the convergence of the MM algorithm when calculating ${\bf F}_\text{RFR}^{(k_\text{o})}$ in step 14  can be proved with the following inequalities 
   \begin{equation}\label{inequality2}
   	\begin{split}	
   J\left({\bf F}_\text{RFR}^{(k_\text{i}+1)}; {\bf Y}_\text{R}^{(k_\text{o}+1)} \right) &\le \tilde J\left({\bf F}_\text{RFR}^{(k_\text{i}+1)}; {\bf Y}_\text{R}^{(k_\text{o}+1)},{\bf F}_\text{RFR}^{(k_\text{i})} \right) \\&\le \tilde J\left({\bf F}_\text{RFR}^{(k_\text{i})}; {\bf Y}_\text{R}^{(k_\text{o}+1)},{\bf F}_\text{RFR}^{(k_\text{i})} \right)\\&=J\left({\bf F}_\text{RFR}^{(k_\text{i})}; {\bf Y}_\text{T}^{(k_\text{o}+1)} \right),  
   	\end{split}	
   \end{equation}
   where  ${\bf Y}_\text{R}^{(k_\text{o}+1)} = {\bf F}_\text{RFT}^{(k_\text{o}+1)} {\bf F}_\text{BBT}^{H(k_\text{o}+1)} {\bf F}_\text{BBR}^{H(k_\text{o}+1)}$.
   After the converge of step 14, we obtain $J({\bf F}_\text{RFR}^{(k_\text{o}+1)}; {\bf Y}_\text{R}^{(k_\text{o}+1)} ) \le J({\bf F}_\text{RFR}^{(k_\text{o})}; {\bf Y}_\text{R}^{(k_\text{o}+1)} )$.  Then, based on  above observations, we have 
   \begin{equation}
   	\begin{split}
   & J\left({\bf F}_\text{RFT}^{(k_\text{o})},{\bf F}_\text{BBT}^{(k_\text{o})}, {\bf F}_\text{BBR}^{(k_\text{o})}, {\bf F}_\text{RFR}^{(k_\text{o})} \right) =J\left({\bf F}_\text{RFR}^{(k_\text{o})};{\bf F}_\text{BBT}^{(k_\text{o})} {\bf F}_\text{BBR}^{H(k_\text{o})} {\bf F}_\text{RFR}^{H(k_\text{o})} \right)	\\
   &\qquad \qquad\qquad \qquad\quad \mathop  \ge \limits^{(a)}  J\left({\bf F}_\text{RFT}^{(k_\text{o})};{\bf F}_\text{BBT}^{(k_\text{o}+1)} {\bf F}_\text{BBR}^{H(k_\text{o})} {\bf F}_\text{RFR}^{H(k_\text{o})} \right)\\
   &\qquad \qquad\qquad\quad  \qquad\mathop  \ge \limits^{(b)}  J\left({\bf F}_\text{RFT}^{(k_\text{o}+1)};{\bf F}_\text{BBT}^{(k_\text{o}+1)} {\bf F}_\text{BBR}^{H(k_\text{o})} {\bf F}_\text{RFR}^{H(k_\text{o})} \right)\\ 
   &\qquad \qquad\qquad\quad  \qquad=  J\left({\bf F}_\text{RFR}^{(k_\text{o})}; {\bf F}_\text{RFT}^{(k_\text{o}+1)} {\bf F}_\text{BBT}^{H(k_\text{o}+1)} {\bf F}_\text{BBR}^{H(k_\text{o})}\right)\\     
   &\qquad \qquad\qquad\quad  \qquad\mathop  \ge \limits^{(c)}  J\left({\bf F}_\text{RFR}^{(k_\text{o})}; {\bf F}_\text{RFT}^{(k_\text{o}+1)} {\bf F}_\text{BBT}^{H(k_\text{o}+1)} {\bf F}_\text{BBR}^{H(k_\text{o}+1)}\right)\\
   &\qquad \qquad\qquad\quad  \qquad\mathop  \ge \limits^{(d)}  J\left({\bf F}_\text{RFR}^{(k_\text{o}+1)}; {\bf F}_\text{RFT}^{(k_\text{o}+1)} {\bf F}_\text{BBT}^{H(k_\text{o}+1)} {\bf F}_\text{BBR}^{H(k_\text{o}+1)}\right)\\   
   &\qquad \qquad\qquad\quad  \qquad=	J\left({\bf F}_\text{RFT}^{(k_\text{o}+1)},{\bf F}_\text{BBT}^{(k_\text{o}+1)}, {\bf F}_\text{BBR}^{(k_\text{o}+1)}, {\bf F}_\text{RFR}^{(k_\text{o}+1)} \right),	\\ 
   	\end{split}	
   \end{equation}
   where  (a) and (c) respectively follow from 
   \begin{equation}
   	\begin{split}
   J\left({\bf F}_\text{RFT}^{(k_\text{o})};{\bf F}_\text{BBT}^{(k_\text{o}+1)} \right.& \left.{\bf F}_\text{BBR}^{H(k_\text{o})} {\bf F}_\text{RFR}^{H(k_\text{o})} \right) \\
   & =  \mathop{\min}\limits_{{\bf F}_\text{BBT}} J\left({\bf F}_\text{RFT}^{(k_\text{o})};{\bf F}_\text{BBT} {\bf F}_\text{BBR}^{H(k_\text{o})} {\bf F}_\text{RFR}^{H(k_\text{o})} \right),
   	\end{split}	
   \end{equation}
   and 
   \begin{equation}
   	\begin{split}
   J\left({\bf F}_\text{RFR}^{(k_\text{o})}; {\bf F}_\text{RFT}^{(k_\text{o}+1)} \right.& \left.{\bf F}_\text{BBT}^{H(k_\text{o}+1)} {\bf F}_\text{BBR}^{H(k_\text{o}+1)}\right) \\
   &= \mathop{\min}\limits_{{\bf F}_\text{BBR}}  J\left({\bf F}_\text{RFR}^{(k_\text{o})}; {\bf F}_\text{RFT}^{(k_\text{o}+1)} {\bf F}_\text{BBT}^{H(k_\text{o}+1)} {\bf F}_\text{BBR}^H\right),
   	\end{split}	
   \end{equation}
   (b) and (d)  follow from \eqref{inequality1} and \eqref{inequality2}, respectively.
   
   Thus, $\{J({\bf F}_\text{RFT}^{(k_\text{o})},{\bf F}_\text{BBT}^{(k_\text{o})}, {\bf F}_\text{BBR}^{(k_\text{o})}, {\bf F}_\text{RFR}^{(k_\text{o})} ) \}$ is a non-increasing  sequence and thus it converges since  $J({\bf F}_\text{RFT},{\bf F}_\text{BBT}, {\bf F}_\text{BBR}, {\bf F}_\text{RFR})$ is lower bounded. The proof of convergence of Algorithm 2 is completed.
   
   For Algorithm 2,  the main complexity in each iteration includes the following  three parts: 
   
   1)  Compute ${\bf F}_\text{BBT}$ and ${\bf F}_\text{BBR}$. The  complexity  of  pseudo inversion  can be measured by  the complexity of singular value decomposition. Thus, the main complexity for this part is $\mathcal{O}( n_\text{t}N_\text{RFT}^2 +n_\text{r} N_\text{RFR}^2+ n_\text{s}^2( n_\text{r}+n_\text{t} )   )$. 
   
   2) Compute  ${\bf F}_\text{RFT}$ and ${\bf F}_\text{RFR}$ with MM methods. The main complexity comes  from finding the maximum eigenvalue of ${\bf Q}_\text{T}$ and the maximum eigenvalue of ${\bf Q}_\text{R}$.  The main complexity of this part is $\mathcal{O}( (n_\text{r}N_\text{RFR})^3+(n_\text{t}N_\text{RFT})^3  )$.
   
   Thus, the main complexity for Algorithm 2 is given by $\mathcal{O}( K_\text{out}(K_\text{in}( (n_\text{r}N_\text{RFR})^3+(n_\text{t}N_\text{RFT})^3  )+  n_\text{t}N_\text{RFT}^2 +n_\text{r} N_\text{RFR}^2+ n_\text{s}^2( n_\text{r}+n_\text{t} )   )    ) $.

\end{appendices}

\bibliographystyle{ieeetr}
\bibliography{ref}

\begin{thebibliography}{10}

\bibitem{Ming2017xiao}
M.~{Xiao}, S.~{Mumtaz}, Y.~{Huang}, L.~{Dai}, Y.~{Li}, M.~{Matthaiou}, G.~K.
  {Karagiannidis}, E.~{Björnson}, K.~{Yang}, C.~{I}, and A.~{Ghosh},
  ``Millimeter wave communications for future mobile networks,'' {\em IEEE
  Journal on Selected Areas in Communications}, vol.~35, pp.~1909--1935, Sep.
  2017.

\bibitem{zhang2019precoding}
Y.~Zhang, M.~Xiao, S.~Han, M.~Skoglund, and W.~Meng, ``{On precoding and energy
  efficiency of full-duplex millimeter-wave relays},'' {\em IEEE Transactions
  on Wireless Communications}, vol.~18, no.~3, pp.~1943--1956, 2019.

\bibitem{Wang2018}
X.~{Wang}, L.~{Kong}, F.~{Kong}, F.~{Qiu}, M.~{Xia}, S.~{Arnon}, and G.~{Chen},
  ``Millimeter wave communication: A comprehensive survey,'' {\em IEEE
  Communications Surveys Tutorials}, vol.~20, pp.~1616--1653, thirdquarter
  2018.

\bibitem{el2014spatially}
O.~El~Ayach, S.~Rajagopal, S.~Abu-Surra, Z.~Pi, and R.~W. Heath, ``{Spatially
  sparse precoding in millimeter wave MIMO systems},'' {\em IEEE transactions
  on wireless communications}, vol.~13, no.~3, pp.~1499--1513, 2014.

\bibitem{rappaport2017overview}
T.~S. Rappaport, Y.~Xing, G.~R. MacCartney, A.~F. Molisch, E.~Mellios, and
  J.~Zhang, ``{Overview of millimeter wave communications for fifth-generation
  (5G) wireless networks—With a focus on propagation models},'' {\em IEEE
  Transactions on Antennas and Propagation}, vol.~65, no.~12, pp.~6213--6230,
  2017.

\bibitem{atzeni2017full}
I.~Atzeni and M.~Kountouris, ``{Full-duplex MIMO small-cell networks with
  interference cancellation},'' {\em IEEE Transactions on Wireless
  Communications}, vol.~16, no.~12, pp.~8362--8376, 2017.

\bibitem{han2018precoding}
S.~Han, Y.~Zhang, W.~Meng, and Z.~Zhang, ``{Precoding design for full-duplex
  transmission in millimeter wave relay backhaul},'' {\em Mobile Networks and
  Applications}, vol.~23, no.~5, pp.~1416--1426, 2018.

\bibitem{sharma2017dynamic}
S.~K. Sharma, T.~E. Bogale, L.~B. Le, S.~Chatzinotas, X.~Wang, and
  B.~Ottersten, ``{Dynamic spectrum sharing in 5G wireless networks with
  full-duplex technology: Recent advances and research challenges},'' {\em IEEE
  Communications Surveys \& Tutorials}, vol.~20, no.~1, pp.~674--707, 2017.

\bibitem{dinc2017millimeter}
T.~Dinc and H.~Krishnaswamy, ``{Millimeter-wave full-duplex wireless:
  Applications, antenna interfaces and systems},'' in {\em 2017 IEEE Custom
  Integrated Circuits Conference (CICC)}, pp.~1--8, IEEE, 2017.

\bibitem{satyanarayana2018hybrid}
K.~Satyanarayana, M.~El-Hajjar, P.-H. Kuo, A.~Mourad, and L.~Hanzo, ``{Hybrid
  beamforming design for full-duplex millimeter wave communication},'' {\em
  IEEE Transactions on Vehicular Technology}, vol.~68, no.~2, pp.~1394--1404,
  2018.

\bibitem{he2017spatiotemporal}
Y.~He, X.~Yin, and H.~Chen, ``{Spatiotemporal characterization of
  self-interference channels for 60-GHz full-duplex communication},'' {\em IEEE
  Antennas and Wireless Propagation Letters}, vol.~16, pp.~2220--2223, 2017.

\bibitem{abbas2016full}
H.~Abbas and K.~Hamdi, ``{Full duplex relay in millimeter wave backhaul
  links},'' in {\em 2016 IEEE Wireless Communications and Networking
  Conference}, pp.~1--6, IEEE, 2016.

\bibitem{xiao2017full}
Z.~Xiao, P.~Xia, and X.-G. Xia, ``{Full-duplex millimeter-wave
  communication},'' {\em IEEE Wireless Communications}, vol.~24, no.~6,
  pp.~136--143, 2017.

\bibitem{huang2018deep}
H.~Huang, J.~Yang, H.~Huang, Y.~Song, and G.~Gui, ``{Deep learning for
  super-resolution channel estimation and DOA estimation based massive MIMO
  system},'' {\em IEEE Transactions on Vehicular Technology}, vol.~67, no.~9,
  pp.~8549--8560, 2018.

\bibitem{long2018data}
Y.~Long, Z.~Chen, J.~Fang, and C.~Tellambura, ``{Data-driven-based analog beam
  selection for hybrid beamforming under mm-wave channels},'' {\em IEEE Journal
  of Selected Topics in Signal Processing}, vol.~12, no.~2, pp.~340--352, 2018.

\bibitem{samuel2017deep}
N.~Samuel, T.~Diskin, and A.~Wiesel, ``{Deep MIMO detection},'' in {\em 2017
  IEEE 18th International Workshop on Signal Processing Advances in Wireless
  Communications (SPAWC)}, pp.~1--5, IEEE, 2017.

\bibitem{elbir2019cnn}
A.~M. Elbir, ``{CNN-based precoder and combiner design in mmWave MIMO
  systems},'' {\em IEEE Communications Letters}, vol.~23, no.~7,
  pp.~1240--1243, 2019.

\bibitem{huang2019deep}
H.~Huang, Y.~Song, J.~Yang, G.~Gui, and F.~Adachi, ``{Deep-learning-based
  millimeter-wave massive MIMO for hybrid precoding},'' {\em IEEE Transactions
  on Vehicular Technology}, vol.~68, no.~3, pp.~3027--3032, 2019.

\bibitem{lin2019beamforming}
T.~Lin and Y.~Zhu, ``{Beamforming design for large-scale antenna arrays using
  deep learning},'' {\em IEEE Wireless Communications Letters}, 2019.

\bibitem{bao2020deep}
X.~Bao, W.~Feng, J.~Zheng, and J.~Li, ``{Deep CNN and Equivalent Channel Based
  Hybrid Precoding for mmWave Massive MIMO Systems},'' {\em IEEE Access},
  vol.~8, pp.~19327--19335, 2020.

\bibitem{elbir2019joint}
A.~M. Elbir and K.~V. Mishra, ``{Joint Antenna Selection and Hybrid Beamformer
  Design using Unquantized and Quantized Deep Learning Networks},'' {\em IEEE
  Transactions on Wireless Communications}, 2019.

\bibitem{yu2016alternating}
X.~Yu, J.-C. Shen, J.~Zhang, and K.~B. Letaief, ``{Alternating minimization
  algorithms for hybrid precoding in millimeter wave MIMO systems},'' {\em IEEE
  Journal of Selected Topics in Signal Processing}, vol.~10, no.~3,
  pp.~485--500, 2016.

\bibitem{lin2019hybrid}
T.~Lin, J.~Cong, Y.~Zhu, J.~Zhang, and K.~B. Letaief, ``{Hybrid beamforming for
  millimeter wave systems using the MMSE criterion},'' {\em IEEE Transactions
  on Communications}, vol.~67, no.~5, pp.~3693--3708, 2019.

\bibitem{sohrabi2016hybrid}
F.~Sohrabi and W.~Yu, ``{Hybrid digital and analog beamforming design for
  large-scale antenna arrays},'' {\em IEEE Journal of Selected Topics in Signal
  Processing}, vol.~10, no.~3, pp.~501--513, 2016.

\bibitem{lee2014af}
J.~Lee and Y.~H. Lee, ``{AF relaying for millimeter wave communication systems
  with hybrid RF/baseband MIMO processing},'' in {\em 2014 IEEE International
  Conference on Communications (ICC)}, pp.~5838--5842, IEEE, 2014.

\bibitem{tsinos2018hybrid}
C.~G. Tsinos, S.~Chatzinotas, and B.~Ottersten, ``{Hybrid analog-digital
  transceiver designs for mmwave amplify-and-forward relaying systems},'' in
  {\em 2018 41st International Conference on Telecommunications and Signal
  Processing (TSP)}, pp.~1--6, IEEE, 2018.

\bibitem{boyd}
S.~Boyd, N.~Parikh, E.~Chu, B.~Peleato, J.~Eckstein, {\em et~al.},
  ``Distributed optimization and statistical learning via the alternating
  direction method of multipliers,'' {\em Foundations and Trends in Machine
  learning}, vol.~3, no.~1, pp.~1--122, 2011.

\bibitem{sun2016majorization}
Y.~Sun, P.~Babu, and D.~P. Palomar, ``{Majorization-minimization algorithms in
  signal processing, communications, and machine learning},'' {\em IEEE
  Transactions on Signal Processing}, vol.~65, no.~3, pp.~794--816, 2016.

\bibitem{wu2017transmit}
L.~Wu, P.~Babu, and D.~P. Palomar, ``{Transmit waveform/receive filter design
  for MIMO radar with multiple waveform constraints},'' {\em IEEE Transactions
  on Signal Processing}, vol.~66, no.~6, pp.~1526--1540, 2017.

\bibitem{arora2019hybrid}
A.~Arora, C.~G. Tsinos, S.~Chatzinotas, B.~Ottersten, {\em et~al.}, ``{Hybrid
  Transceivers Design for Large-Scale Antenna Arrays Using
  Majorization-Minimization Algorithms},'' {\em IEEE Transactions on Signal
  Processing}, 2019.

\bibitem{comon1990tracking}
P.~Comon and G.~H. Golub, ``{Tracking a few extreme singular values and vectors
  in signal processing},'' {\em Proceedings of the IEEE}, vol.~78, no.~8,
  pp.~1327--1343, 1990.

\bibitem{ITP98}
P.~L. {Bartlett}, ``The sample complexity of pattern classification with neural
  networks: the size of the weights is more important than the size of the
  network,'' {\em IEEE Transactions on Information Theory}, vol.~44,
  pp.~525--536, March 1998.

\bibitem{17}
K.~Hornik, M.~Stinchcombe, H.~White, {\em et~al.}, ``Multilayer feedforward
  networks are universal approximators.,'' {\em Neural networks}, vol.~2,
  no.~5, pp.~359--366, 1989.

\bibitem{yy}
Y.~Ye, M.~Xiao, and M.~Skoglund, ``Decentralized multi-task learning based on
  extreme learning machines,'' {\em arXiv preprint arXiv:1904.11366}, 2019.

\bibitem{18}
T.~D. Sanger, ``Optimal unsupervised learning in a single-layer linear
  feedforward neural network,'' {\em Neural networks}, vol.~2, no.~6,
  pp.~459--473, 1989.

\bibitem{20}
G.-B. Huang, Q.-Y. Zhu, and C.-K. Siew, ``Extreme learning machine: theory and
  applications,'' {\em Neurocomputing}, vol.~70, no.~1-3, pp.~489--501, 2006.

\bibitem{21}
G.-B. Huang, H.~Zhou, X.~Ding, and R.~Zhang, ``Extreme learning machine for
  regression and multiclass classification,'' {\em IEEE Transactions on
  Systems, Man, and Cybernetics, Part B (Cybernetics)}, vol.~42, no.~2,
  pp.~513--529, 2011.

\bibitem{22}
G.-B. Huang, D.~H. Wang, and Y.~Lan, ``Extreme learning machines: a survey,''
  {\em International journal of machine learning and cybernetics}, vol.~2,
  no.~2, pp.~107--122, 2011.

\end{thebibliography}
\end{document}